\documentclass[journal]{IEEEtran}
\IEEEoverridecommandlockouts
\usepackage{placeins}
\usepackage{amsthm}
\usepackage{amsmath}
\usepackage{mathtools}
\usepackage{amssymb}
\usepackage{algorithmic}
\usepackage{algorithm}
\usepackage{verbatim}
\usepackage{cite}
\usepackage{url}
\usepackage{cases}
\usepackage{MnSymbol}
\usepackage[all,cmtip]{xy}
\usepackage{esint}
\usepackage{ulem}
\usepackage{stackrel}
\usepackage{calc}
\usepackage{units}
\usepackage{nicefrac}
\usepackage{enumerate}
\usepackage{graphicx}
\usepackage{subcaption}
\usepackage{float}
\usepackage[table,xcdraw]{xcolor}

\newtheorem{thm}{Theorem}

\newtheorem{prop}{Proposition}
\newtheorem{lem}{Lemma}

\IEEEoverridecommandlockouts
\begin{document}
    
    \title{Hybrid Analog-Digital Beamforming for Massive MIMO Systems}
    
    \author{\thanks{
            This project has received funding from the European Union's Horizon 2020 research and innovation program under grant agreement No. 646804-ERC-COG-BNYQ, and from the Israel Science Foundation under Grant no. 335/14.}
        Shahar Stein, \emph{Student
            IEEE} and Yonina C. Eldar, \emph{Fellow IEEE} }
    
    \maketitle
    
    \IEEEpeerreviewmaketitle{}

    \begin{abstract}
	In massive MIMO systems, hybrid beamforming is an essential technique for exploiting the potential array gain without using a dedicated RF chain for each antenna.
	In this work, we consider the data phase in a massive MIMO communication process, where the transmitter and receiver use fewer RF chains than antennas.
	We examine several different fully- and partially connected schemes and consider the design of hybrid beamformers that minimize the estimation error in the data.
	For the hybrid precoder, we introduce a framework for approximating the optimal fully-digital precoder with a feasible hybrid one. 
	We exploit the fact that the fully-digital precoder is unique only up to a unitary matrix and optimize over this matrix and the hybrid precoder alternately. Our alternating minimization of approximation gap (Alt-MaG) framework improves the performance over state-of-the-art methods with no substantial increase in complexity. 
	In addition, we present a special case of Alt-MaG, minimal gap iterative quantization (MaGiQ), that results in low complexity and lower mean squared error (MSE) than other common methods, in the case of very few RF chains.
	MaGiQ is shown to coincide with the optimal fully-digital solution in some scenarios.
	For combiner design, we exploit the structure of the MSE objective and develop a greedy ratio trace maximization technique, that achieves low MSE under various settings. 
	All of our algorithms can be used with multiple hardware architectures.
\end{abstract}

\section{Introduction}\label{sec:Intro}

Massive MIMO wireless systems have emerged as a leading candidate for 5G wireless access \cite{andrews_what_2014,rusek_scaling_2013}. Along with mmWave technologies, that were recently recognized as essential for coping with the spectrum crunch \cite{rappaport_millimeter_2013}, it offers higher data rates and capacities than traditional MIMO systems. The use of large-scale antenna arrays at both the transmitter and receiver holds the potential for higher array gain than before. To utilize this gain, precoding and combining techniques are used. Traditionally implemented in the baseband (BB), these methods require a dedicated RF hardware per antenna. Unfortunately, when taking a massive amount of antennas into account, this results in a huge computational load and cost, as the RF components are expensive and have high power consumption, especially for mmWave technologies. Hence, it is desirable to design economical hardware that will utilize the potential gain from a large number of cheap antenna elements using a small number of expensive RF chains.

To achieve this goal, several hybrid analog-digital schemes have been suggested \cite{ayach_spatially_2014,alkhateeb_hybrid_2013,lee_hybrid_2015,han_large-scale_2015,mendez-rial_hybrid_2016,kim_mse-based_2015,yu_alternating_2016,singh_feasibility_2015,yu_partially-connected_2017,park_dynamic_2017}. In hybrid precoding and combining, the operations are split between the digital and analog domains: at the transmitter side a low dimensional digital precoder operates on the transmitted signal at BB. An analog precoder then maps the small number of digital outputs to a large number of antennas, and the same is performed at the receiver side. Common analog architectures are based on analog phase shifters and switches. The specific schemes can vary according to the power and area budget. Two main families of architectures are the fully- and partially connected structures. 
Fully connected networks offer a mapping from each antenna to each RF chain and allow to maximize the precoding and combining gain. In the partially connected scheme, a reduced number of analog components is used. This degrades the achieved gain but offers lower power consumption and hardware complexity. 

In the data phase of each coherence interval of the MIMO communication process, the transmitter sends multiple data streams to the receiver over the constant and known channel, using precoding techniques. At the receiver side, a combiner is used to estimate the data vector from the received signal at the antennas. The precoder and combiner are chosen to optimize some desired performance measure, such as estimation error or spectral efficiency of the system.      
Unlike the fully-digital case, when considering a hybrid beamformer, the precoder and combiner matrices cannot have arbitrary entries, but are constrained according to the specific hardware choice.  
For example, when using a phase shifter network at the analog side, only unimodular matrices for the analog beamformer are considered. The goal then is to optimize the performance measure over all pairs of digital and analog precoder matrices, which yields a non-convex difficult optimization problem.

The majority of past works considered the fully-connected phase shifter network. It was suggested in \cite{ayach_spatially_2014} to separate the joint precoder and combiner design problem into two subproblems and solve each independently. Although this approach is sub-optimal, it greatly simplifies the difficult joint optimization problem. For the precoder, it was shown that minimizing the gap between the hybrid precoder and the optimal fully-digital one over all hybrid precoders, approximately leads to the maximization of the system's spectral efficiency. On the combiner side, minimizing the mean squared error (MSE) over all hybrid combiners was shown to be equivalent to minimizing the weighted approximation gap between the fully-digital combiner and the hybrid one, with the weights given by the received signal correlation. Methods for solving these approximation problems were suggested in \cite{ayach_spatially_2014,alkhateeb_hybrid_2013,lee_hybrid_2015,kim_mse-based_2015,mendez-rial_hybrid_2016,yu_alternating_2016,yu_partially-connected_2017}.

The works \cite{ayach_spatially_2014,alkhateeb_hybrid_2013,lee_hybrid_2015} considered precoder design and aimed at maximizing the system's spectral efficiency. They exploit the mmWave sparse multipath channel structure and deduce that the optimal fully-digital precoder is composed of a small sum of steering vectors. They then suggest a variant of the orthogonal matching pursuit (OMP) \cite{CSBook} algorithm to construct a feasible precoder that approximates the optimal one using a dictionary of steering vectors. This solution greatly reduces the problem complexity but results in a large performance gap from the optimal precoder due to restricting the space of possible precoder vectors to steering vectors. 

Similar approaches were taken in \cite{ayach_spatially_2014,kim_mse-based_2015,mendez-rial_hybrid_2016} for the combiner design, under a MSE or spectral efficiency objective.
In \cite{yu_alternating_2016}, two algorithms were suggested to approximate the optimal fully-digital precoder with a feasible one. The first, MO-AltMin, is based on manifold optimization. In each iteration of the algorithm, it assumes a given digital precoder and develops a conjugate gradient method to find an analog precoder that is a local minimizer of the approximation gap from the fully-digital one. Next, the digital precoder is computed using a least squares solution. This method achieves good performance but suffers from high complexity and run time, and is only suitable for fully-connected phase shifter networks. The second approach is a low complexity algorithm that assumes the digital precoder to be a scaled unitary matrix and uses this assumption to produce an upper bound on the approximation gap, which is then minimized over all analog precoders. However, limiting the combiner to such a structure results in performance loss.

For the partially connected architecture, most of the existing works concentrated on fixed sub-arrays where each RF chain is connected to a predetermined sub-array.
In \cite{singh_feasibility_2015}, the authors suggested a low complexity codebook design producing a small dictionary of feasible precoding vectors, that are chosen based on the transmitted signal strength. They then exhaustively search over all possible combinations from the small dictionary to maximize the mutual information between the receiver and transmitter. This last step can result in heavy computational load when the channel is not sparse.
In \cite{yu_alternating_2016}, the authors consider disjoint sub-arrays and suggest an iterative method to approximate the optimal fully-digital precoder, which optimizes over the analog and digital combiners alternately. It is shown that for disjoint arrays, the analog precoder problem is separable in the antennas, and has a closed form solution. For the digital precoder, a semidefinite relaxation is suggested. However, this is both computationally heavy and unnecessary since a closed form solution for the digital combiner is available. This alternating minimization approach still results in large performance gap.
Some less restrictive schemes were considered in \cite{yu_partially-connected_2017,park_dynamic_2017}. In \cite{yu_partially-connected_2017}, a double phase shifter with dynamic mapping is considered. The use of two phase shifters per antenna relaxes the unit modulo constraint that limits the previous solutions, but costs twice the power consumption, area and complexity. In \cite{park_dynamic_2017}, a dynamic sub-array approach is considered for orthogonal frequency-division multiplexing (OFDM), and a greedy algorithm to optimize the array partition based on the long-term channel characteristics is suggested. However, this method is relevant only for OFDM transmissions.

Most of the above works concentrated on specific channel models and hardware architectures and suggested a tailored algorithm for the chosen scenario. In this work, we develop a general framework suitable for various channels and hardware: full and partial networks with different amounts of phase shifters and switches. We consider the data estimation problem in a single user massive MIMO system where both the receiver and transmitter are equipped with large antenna arrays and fewer RF chains than antennas.
We assume a Bayesian model where the data and interference are both random, and aim at minimizing the MSE of the transmitted data from the low dimensional received digital signal at the receiver, over all hybrid precoders and combiners, assuming a fixed number of RF chains. Like previous works, we relax the difficult joint optimization to two separate problems in the precoder and combiner. 

To design the precoder, we present a framework for approximating the optimal fully-digital precoder with a feasible hybrid one. Our alternating minimization of approximation gap (Alt-MaG) method, exploits the fact that there exists an infinite set of optimal precoders, which differ by a unitary matrix. We suggest optimizing over this matrix and the hybrid precoder alternately, to find the fully-digital solution that results in the smallest approximation gap from its hybrid decomposition. 
For the hybrid precoder optimization step, any of the previously suggested methods \cite{alkhateeb_hybrid_2013,ayach_spatially_2014,kim_mse-based_2015,lee_hybrid_2015,mendez-rial_hybrid_2016,yu_alternating_2016,yu_partially-connected_2017} may be used, according to the hardware constraints. By optimizing over the unitary matrix as well, Alt-Mag achieves additional reduction in MSE compared to  state-of-the-art algorithms, with no significant increase in complexity.

We then present a simple possible solution for to hybrid precoder optimization step, that results in a low complexity algorithm, termed minimal gap iterative quantization (MaGiQ). In each iteration of MaGiQ, the fully-digital solution is approximated using the analog precoder alone. This results in a closed form solution, given by a simple quantization function which depends on the hardware structure. We demonstrate in simulations that MaGiQ achieves lower MSE than other low complexity algorithms when using very few RF chains. In addition, in some specific cases, it coincides with the optimal fully-digital solution.

Next, we show that MaGiQ can also be used for the combiner design with mild adjustments. We then suggest an additional greedy ratio trace maximization (GRTM) algorithm that directly minimizes the estimation error using a suitable dictionary that is chosen according to the hardware scheme. In simulations, we demonstrate that GRTM enjoys good performance and short running time, especially when the number of RF chains increases.

The rest of this paper is organized as follows: In Section~\ref{sec:SignalModel} we introduce the signal model and problem formulation, and review common combiner hardware schemes. 
In Section~\ref{sec:DataEst} we introduce the massive MIMO data estimation problem and derive the relevant MSE minimization over the precoder and combiner. 
Sections~\ref{sec:Precoder} and \ref{sec:Combiner} consider the precoder and combiner design problems, respectively. In Section~\ref{sec:simulations} we evaluate the proposed algorithms using numerical experiments under different scenarios. 

The following notations are used throughout the paper: boldface upper-case $\boldsymbol{X}$ and lower-case $\boldsymbol{x}$ letters are used to denote matrices and vectors respectively, and non-bold letters $x$ are scalars.  
The $i$th element in the $j$th column of $\boldsymbol{X}$ is $\left[\boldsymbol{X}\right]_{ij}$. We let $\boldsymbol{I}$ denote the identity matrix of suitable size, $\boldsymbol{X}^{T},\boldsymbol{X}^{*}$ the transpose and conjugate-transpose of $\boldsymbol{X}$ respectively, $\mathbb{E}\left\{\cdot\right\}$ the expectation and $\Vert\cdot\Vert_{p}^2$, $\Vert\cdot\Vert_{F}^2$ the $\ell_p$ and Forbenius norms. The determinant of $\boldsymbol{X}$ is $\left|\boldsymbol{X}\right|$, and $\mathcal{CN}\left(\boldsymbol{x},\boldsymbol{X}\right)$ is the complex-Gaussian distribution with mean $\boldsymbol{x}$ and covariance matrix $\boldsymbol{X}$. The real part of a variable is denoted as $\Re\left\{\cdot\right\}$, $\mathcal{R}\left(\boldsymbol{X}\right)$ is the range space of $\boldsymbol{X}$ and $\boldsymbol{P}_{\boldsymbol{X}}$ is the orthogonal projection onto $\mathcal{R}\left(\boldsymbol{X}\right)$. 

\section{Problem Formulation}\label{sec:SignalModel}

\subsection{Signal Model}
Consider a single user massive MIMO system in which a transmitter with $N_t$ antennas communicates $N_s$ independent data streams using $N_{RF}^t$ RF chains, $N_s \leq N_{RF}^{t}\leq N_{t}$, to a receiver equipped with $N_r$ antennas and $N_{RF}^r$ RF chains, $N_s \leq N_{RF}^{r}\leq N_{r}$. At the transmitter, the RF chains are followed by a network of switches and phase shifters that expands the $N_{RF}$ digital outputs to $N_{t}$ precoded analog signals which feed the transmit antennas. Similarly, at the receiver, the antennas are followed by a network of switches and phase shifters that feed the $N_{RF}^{r}$ RF chains. The specific architecture of the analog hardware at each end can vary according to budget constraints. Some possible choices are presented in the next subsection. For simplicity, we assume in this paper that $N_{RF}^{t}=N_{RF}^{r}=N_{s}$, which corresponds to the minimal possible number of chains, and hence the worst-case scenario. 

The hybrid architecture enables the transmitter to apply a BB precoder $\boldsymbol{F}_{BB}\in\mathbb{C}^{N_{RF}^{t}\times N_{RF}^{t}}$, followed by a RF precoder $\boldsymbol{F}_{RF}\in\mathcal{F}^{N_{t}\times N_{RF}^{t}}$. The properties of the set $\mathcal{F}$ are determined by the specific hardware scheme in use. For example, 
in fully-connected phase shifter networks, $\mathcal{F}$ is the set of unimodular matrices. 
The transmitter obeys a total power constraint such that $\Vert\boldsymbol{F}_{RF}\boldsymbol{F}_{BB}\Vert_{F}^{2}=N_s$. 

The discrete-time $N_{t} \times 1$ transmitted signal can be written as
\begin{equation}
\boldsymbol{x}=\boldsymbol{F}_{RF}\boldsymbol{F}_{BB}\boldsymbol{s},
\end{equation}  
with $\boldsymbol{s}$  the $N_{s}\times1$ symbol vector. We assume, without loss of generality, that $\boldsymbol{s}\sim\mathcal{N}\left(\boldsymbol{0},\boldsymbol{I}_{N_{s}}\right)$. The signal is transmitted over a narrowband block-fading propagation channel such that the $N_r\times1$ received analog signal vector $\boldsymbol{y}$ at the receiver's antennas is
\begin{equation}
\boldsymbol{y}=\sqrt{p_r}\boldsymbol{H}\boldsymbol{F}_{RF}\boldsymbol{F}_{BB}\boldsymbol{s}+\boldsymbol{z},
\end{equation} 
where $\boldsymbol{H}$ is the $N_{r}\times N_{t}$ channel matrix, $p_r$ is the received average power and $\boldsymbol{z}$ is an $N_{r}\times1$ interference vector with zero mean and autocorrelation matrix $\mathbb{E}\left[\boldsymbol{z}\boldsymbol{z}^{H}\right]=\boldsymbol{R}_{z}$. Here, $\boldsymbol{z}$ can represent either noise or an interfering signal. We assume full channel state information (CSI), i.e. the matrix $\boldsymbol{H}$ is known. The interference correlation matrix $\boldsymbol{R}_{z}$ is also known and assumed to be full rank.

At the receiver, an analog combiner maps the $N_r$ inputs to $N_{RF}^r$ RF chains that are then processed at baseband using a digital combiner. This yields the discrete signal
\begin{equation}
\boldsymbol{r}=\sqrt{p_r}\boldsymbol{W}_{BB}^{*}\boldsymbol{W}_{RF}^{*}\boldsymbol{H}\boldsymbol{F}_{RF}\boldsymbol{F}_{BB}\boldsymbol{s}+\boldsymbol{W}_{BB}^{*}\boldsymbol{W}_{RF}^{*}\boldsymbol{z},
\end{equation}  
with $\boldsymbol{W}_{RF}\in\mathcal{W}^{N_{r}\times N_{RF}^{r}}$ the analog combining matrix and $\boldsymbol{W}_{BB}$ the $N_{RF}^r\times N_{RF}^r$ digital combining matrix. The properties of the set $\mathcal{W}$ vary according to the specific analog hardware scheme. 

We wish to design $\boldsymbol{W}_{RF},\boldsymbol{W}_{BB},\boldsymbol{F}_{RF},\boldsymbol{F}_{BB}$ to minimize the estimation error of $\boldsymbol{s}$ from $\boldsymbol{r}$, under different hardware constraints, i.e different feasible sets $\mathcal{W,F}$. 
Thus, we consider the problem
\begin{equation}
\begin{aligned}
\underset{_{\boldsymbol{W}_{RF},\boldsymbol{W}_{BB}\boldsymbol{F}_{RF},\boldsymbol{F}_{BB}}}{\min}&
\mathbb{E}\left[\Vert \boldsymbol{s}-\boldsymbol{W}_{BB}^*\boldsymbol{W}_{RF}^*\boldsymbol{y}\Vert_{2}^2\right]&\\
\text{s.t.}:\;\;\;\;\;\;\;\;\;\;\;\;&
\boldsymbol{F}_{RF}\in\mathcal{F}^{N_{t}\times N_{RF}^{t}},\;\;
\Vert\boldsymbol{F}_{RF}\boldsymbol{F}_{BB}\Vert_{F}^{2}\leq N_{s},& \\
&\boldsymbol{W}_{RF}\in\mathcal{W}^{N_{r}\times N_{RF}^{r}}.\;\;\;\;\;\;\;\;\;\;\;\;\;\;\;\;\;\;\;\;\;\;\;\;\;\;\;\;\;\;\;\;\;\;&
\end{aligned}
\label{eq:Opt0}
\end{equation}

\subsection{Precoder and Combiner Hardware Schemes}
We now present some possible hardware schemes and the feasible sets they dictate. We focus on the combiner architecture and the feasible set $\mathcal{W}$, but all the options presented here can be translated to similar precoder architectures and feasible sets $\mathcal{F}$. 
The networks below are based on phase shifters and/or switches, with different connectivity levels. In practice, the choice of scheme can be affected by a variety of budget constraints such as power consumption, price, and area against the requested array gain. For example,
phase shifters allow for a more flexible design than switches, but their power consumption is higher.   
\begin{enumerate}[(S1)]
	\item \textbf{Fully Connected Phase Shifters and Switches Network:} Each antenna is connected to an RF chain through an independent on/off switch followed by a phase shifter.
	%     with a total of $N_r\cdot N_{RF}^r$ phase shifters and switches. 
	All the incoming analog signals from the different antennas are combined before feeding the RF chain. In this architecture, the entries of the analog combining matrix $\boldsymbol{W}_{RF}$ can be either zeros or unit modulus. Hence the set $\mathcal{W}$ is defined by 
	$\mathcal{W}^{N_{r}\times N_{RF}^{r}}=\left\{\boldsymbol{W}\in\mathbb{C}^{N_{r}\times N_{RF}^{r}}:\left|w_{ij}\right|\in\left\{0,1\right\}\right\}$. The orthogonal projection $\boldsymbol{P}_{\mathcal{W}}$ onto the feasible set dictated by this scheme is defined as
	\begin{equation}
	\left[\boldsymbol{P}_{\mathcal{W}}\left(\boldsymbol{A}\right)\right]_{il} = \begin{cases}
	e^{j2\pi{\angle\left[\boldsymbol{A}\right]_{il}}}, &\left|\left[\boldsymbol{A}\right]_{il}\right|\geq \frac{1}{2}\\
	0, &\left|\left[\boldsymbol{A}\right]_{il}\right|< \frac{1}{2}\\
	\end{cases}. \label{eq:Qs2}
	\end{equation}
	The fully-connected phase shifters and switches network is presented in Fig. \ref{fig:S1}.\label{scheme:FullyFirst}
	
	\item \textbf{Fully Connected Phase Shifters Network:} (architecture A1 in \cite{mendez-rial_hybrid_2016}). In this case, each antenna is connected to an RF chain through an independent phase shifter, without a switch. 
	Thus, $\boldsymbol{W}_{RF}$ is a unimodular matrix: $\mathcal{W}^{N_{r}\times N_{RF}^{r}}=\left\{\boldsymbol{W}\in\mathbb{C}^{N_{r}\times N_{RF}^{r}}:\left|w_{ij}\right|=1\right\}$. Here we have
	\begin{equation}
	\left[\boldsymbol{P}_{\mathcal{W}}\left(\boldsymbol{A}\right)\right]_{il} = e^{j2\pi{\angle\left[\boldsymbol{A}\right]_{il}}} .\label{eq:Qs1}
	\end{equation}
	Figure \ref{fig:S2}.\label{scheme:FullyPhaseOnly} demonstrates this setting.
	
	\item \textbf{Switching Network:} (architecture A5 in \cite{mendez-rial_hybrid_2016}). This is an antenna selection scheme, where each RF chain is preceded by a single switch that can toggle between the $N_{r}$ antennas. A total of $N_{RF}^{r}$ antennas are selected and sampled. 
	%    It has a total of $N_{RF}^r$ $N_r$-to-$1$ switches.
	It follows that $\boldsymbol{W}_{RF}$ is partial permutation matrix: $\mathcal{W}^{N_{r}\times N_{RF}^{r}}=\left\{\boldsymbol{W}\in\mathbb{C}^{N_{r}\times N_{RF}^{r}}:w_{ij}\in\left\{0,1\right\},\Vert\boldsymbol{w}_{j}\Vert_{0}=1\right\}$. The projection $\boldsymbol{P}_{\mathcal{W}}\left(\boldsymbol{A}\right)$ chooses the entry with largest absolute value in each column of $\boldsymbol{A}$ and sets the corresponding entry of the output to $1$, while all other entries are set to $0$. 
	This setting is shown in Fig. \ref{fig:S3}. Another structure is the switching network with sub-arrays (architecture A6 in \cite{mendez-rial_hybrid_2016}), where the switching range is limited to a certain subset of antennas for each RF chain. 
	Note that in switching networks, phase shifters are unnecessary, since each input can be processed independently at BB.\label{scheme:Switching}
	
	\item \textbf{Partially Connected Phase Shifters Network with Fixed Sub Arrays:} (architecture A2 in \cite{mendez-rial_hybrid_2016}). This scheme divides the array into (possibly overlapping) sub-arrays of size $G$. Each sub-array operates as a fully-connected phase shifter network with a single RF chain. 
	This results in the feasible set $\mathcal{W}^{N_{r}\times N_{RF}^{r}}=\left\{\boldsymbol{W}\in\mathbb{C}^{N_{r}\times N_{RF}^{r}}:\left|w_{ij}\right|=1,\forall i\in\mathcal{S}_{j},\left|w_{ij}\right|=0,\forall i\notin\mathcal{S}_{j}\right\}$, with $\mathcal{S}_{j}$ the group of indices corresponding to antennas that belong to the $j$th sub-array. 
	In this case $\boldsymbol{P}_{\mathcal{W}}\left(\boldsymbol{A}\right)$ applies the mapping in (\ref{eq:Qs1}) to the $G$ entries of $\boldsymbol{a}_j$ that belong to $\mathcal{S}_{j}$, and sets all others to $0$. The fixed sub-arrays network is presented in Fig. \ref{fig:S4}. 
	The parameter $G$ represents the connectivity factor; for $G=N_{r}$, we get the fully-connected network S\ref{scheme:FullyPhaseOnly}. 
	\label{scheme:PhaseSubArrays}
	
	\item \textbf{Flexible Partially Connected Phase Shifters Network with Sub Arrays:} This architecture is similar to the previous one, except that the contributing antennas to each RF chain are not fixed, but can be optimized, i.e. the subgroups $\mathcal{S}_{j}$ are flexible. 
	This is implemented using a $N_r$-to-$1$ switch that precedes the phase shifters. 
	The corresponding feasible set is $\mathcal{W}^{N_{r}\times N_{RF}^{r}}=\left\{\boldsymbol{W}\in\mathbb{C}^{N_{r}\times N_{RF}^{r}}:\left|w_{ij}\right|\in\left\{0,1\right\},\Vert\boldsymbol{w}_{j}\Vert_{0}=G\right\}$.
	The function $\boldsymbol{P}_{\mathcal{W}}\left(\boldsymbol{A}\right)$ in this case chooses the $G$ entries with largest absolute value in each column of $\boldsymbol{A}$ and applies (\ref{eq:Qs1}) to them while setting all others to $0$.
	\label{scheme:FlexiLast}
\end{enumerate}

The networks above constitute some of the main classes of existing analog architectures. They present different levels of challenge in terms of hardware complexity, as in Table \ref{table:hardwareComlex}. For a complete power consumption comparison, the reader is referred to the example in \cite{mendez-rial_hybrid_2016}. One may define many additional architectures based on the above networks by modifying their components and connectivity factors.

\begin{table*}[]
	\centering
	\caption{Hardware complexity comparison between different architectures}
	\label{table:hardwareComlex}
	\begin{tabular}{|l|m{1.5cm}|m{1.5cm}|m{1.5cm}|m{1.5cm}|m{1.5cm}|}
		\hline
		\rowcolor[HTML]{EFEFEF} 
		& \multicolumn{1}{l|}{\cellcolor[HTML]{EFEFEF}\textbf{S1}} & \multicolumn{1}{l|}{\cellcolor[HTML]{EFEFEF}\textbf{S2}} & \multicolumn{1}{l|}{\cellcolor[HTML]{EFEFEF}\textbf{S3}} & \multicolumn{1}{l|}{\cellcolor[HTML]{EFEFEF}\textbf{S4}} & \multicolumn{1}{l|}{\cellcolor[HTML]{EFEFEF}\textbf{S5}} \\ \hline
		\cellcolor[HTML]{EFEFEF}{\color[HTML]{333333} \textbf{Phase Shifters \#}} & $N_{r}\cdot N_{RF}^r$ & $N_{r}\cdot N_{RF}^r$ & 0 & $G\cdot N_{RF}^r$ & $G\cdot N_{RF}^r$ \\ \hline
		\cellcolor[HTML]{EFEFEF}{\color[HTML]{333333} \textbf{Switches \#}} & $N_{r}\cdot N_{RF}^r$ & 0 & $N_{RF}^r$ & 0 & $G\cdot N_{RF}^r$ \\ \hline
		\cellcolor[HTML]{EFEFEF}{\color[HTML]{333333} \textbf{Switch Type}} & on-off & - & $N_r$-to-$1$ & - & $N_r$-to-$1$ \\ \hline
	\end{tabular}
\end{table*}

\begin{figure}
	\begin{center}
		\begin{subfigure}[a]{0.49\linewidth}
			\centering      
			\vspace{-1cm}          
			\includegraphics[width=0.9\linewidth]{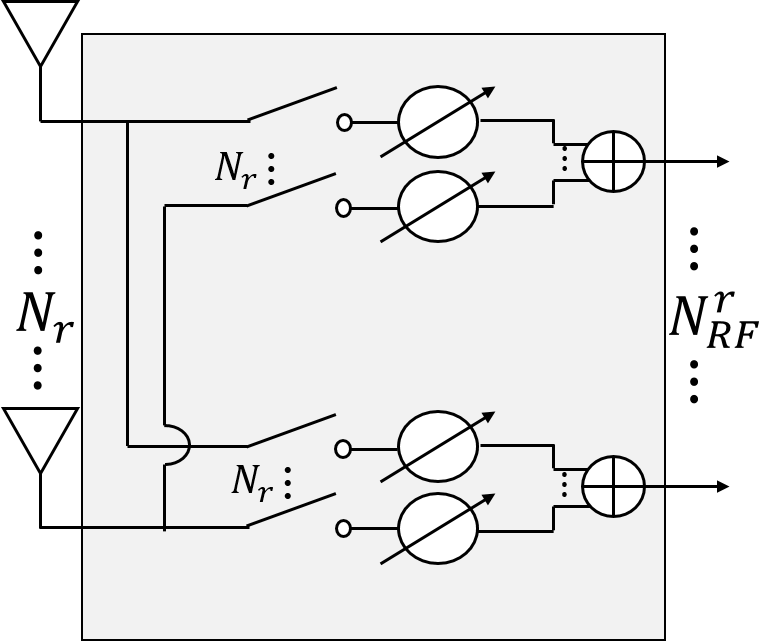}
			\subcaption{}\label{fig:S1}
			\vspace{0.6cm}    
		\end{subfigure}
		\begin{subfigure}[b]{0.49\linewidth}
			\centering      
			\vspace{0.8cm}                
			\includegraphics[width=0.9\linewidth]{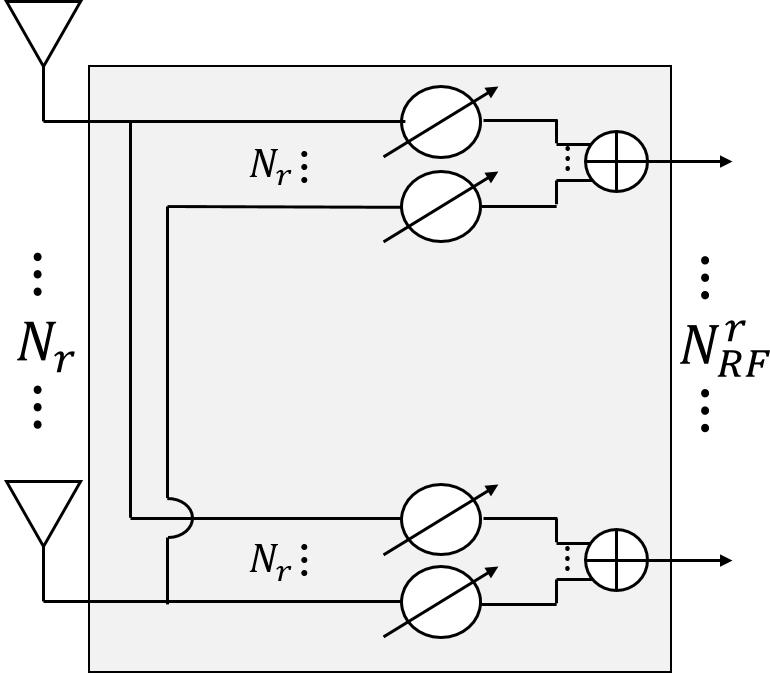}
			\subcaption{}\label{fig:S2}
			\vspace{-1cm}    
		\end{subfigure}    
		\begin{subfigure}[c]{0.49\linewidth}
			\centering            
			\includegraphics[width=0.9\linewidth]{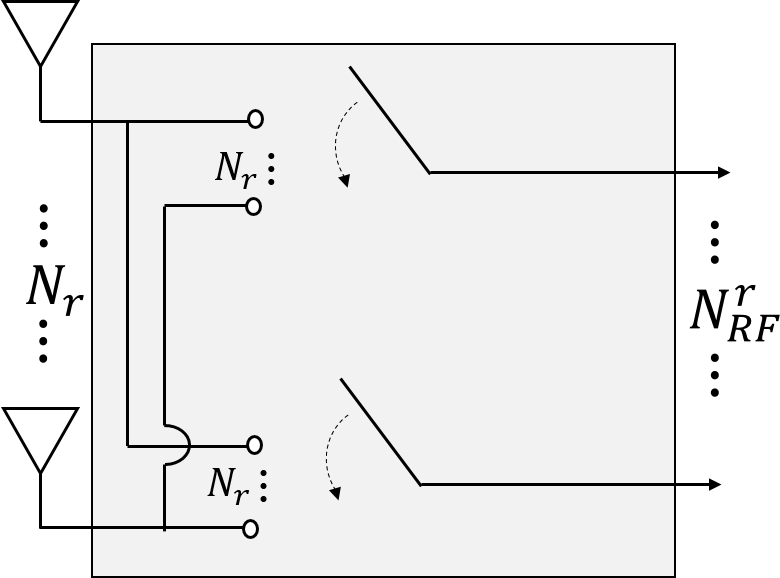}
			\subcaption{}\label{fig:S3}
			\vspace{0.5cm}    
		\end{subfigure}
		\begin{subfigure}[d]{0.49\linewidth}
			\centering            
			\includegraphics[width=0.9\linewidth]{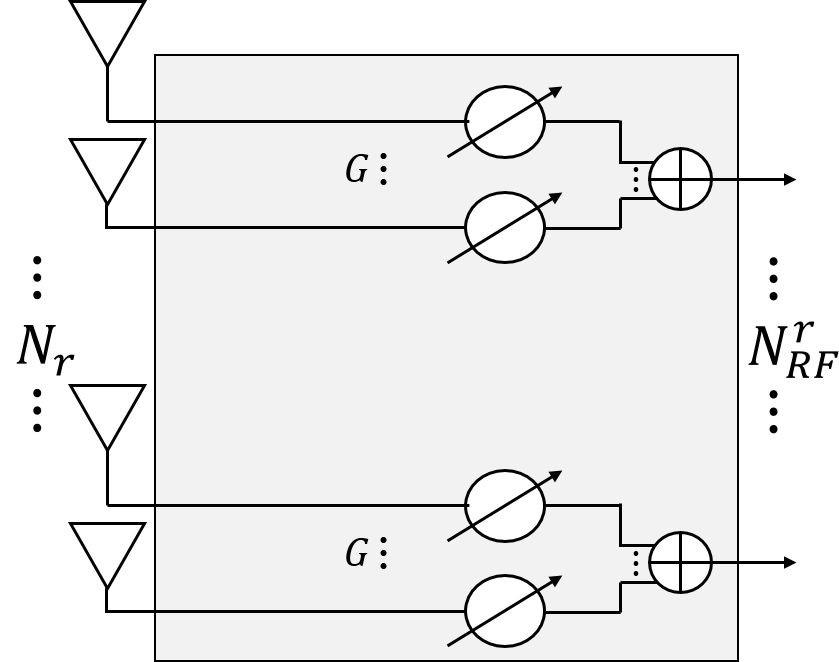}
			\subcaption{}\label{fig:S4}
			\vspace{0.6cm}    
		\end{subfigure}
		\begin{subfigure}[e]{0.49\linewidth}
			\centering            
			\includegraphics[width=0.9\linewidth]{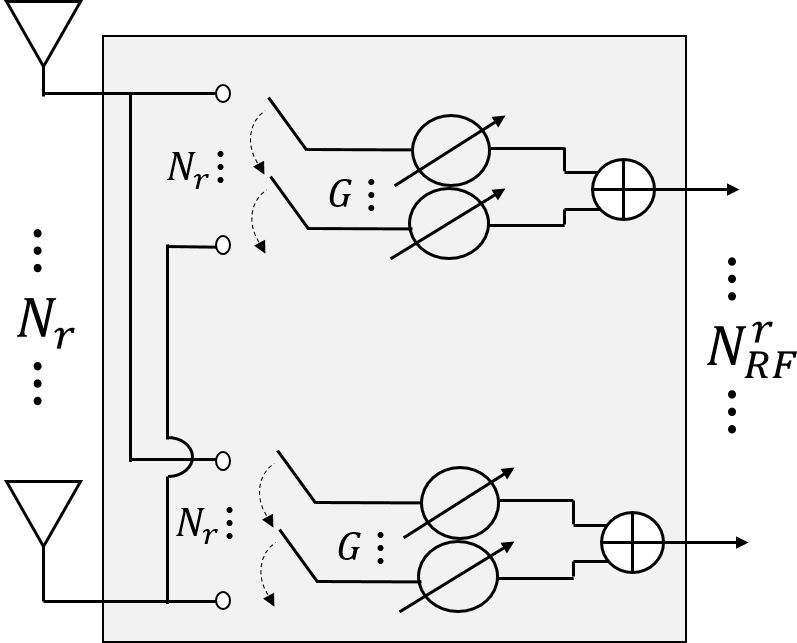}
			\subcaption{}\label{fig:S5}
			%    \vspace{0.6cm}    
		\end{subfigure}
		\caption{Combiner hardware schemes. (\ref{fig:S1}) S\ref{scheme:FullyFirst}: Fully connected phase shifters and switches network, (\ref{fig:S2}) S\ref{scheme:FullyPhaseOnly}: Fully connected phase shifters network, (\ref{fig:S3}) S\ref{scheme:Switching}: Switching network, (\ref{fig:S4}) S\ref{scheme:PhaseSubArrays}: Fixed partially connected phase shifters network with sub-arrays, and (\ref{fig:S5}) S\ref{scheme:FlexiLast}: Flexible partially connected phase shifters network with sub-arrays.}\label{fig:Schemes}
	\end{center}
	\vspace{-0.9cm}
\end{figure}

\section{MSE Minimization}\label{sec:DataEst}        
To optimally design the precoder and combiner we consider the MSE minimization problem in (\ref{eq:Opt0}).

Since there are no constraints on $\boldsymbol{W}_{BB}$ in (\ref{eq:Opt0}), it can be chosen as any $N_{RF}^r\times N_{RF}^r$ matrix, including the linear MMSE estimator of $\boldsymbol{s}$ from the measurements $\tilde{\boldsymbol{y}}=\boldsymbol{W}_{RF}^{*}\boldsymbol{y}$. This estimator depends on the matrices $\boldsymbol{F}_{RF},\boldsymbol{F}_{BB}$,$\boldsymbol{W}_{RF}$ and is given by
\begin{equation}
\begin{aligned}
\boldsymbol{W}_{BB,opt}^{*}=
\bar{\boldsymbol{H}}^{*}\boldsymbol{W}_{RF}\left[\boldsymbol{W}_{RF}{*}\left(\bar{\boldsymbol{H}}\bar{\boldsymbol{H}}^{*}+\boldsymbol{R}_{z}\right)\boldsymbol{W}_{RF}\right]^{-1},
\end{aligned}\label{eq:Wbb}
\end{equation}
with 
\begin{equation}
\bar{\boldsymbol{H}}=\sqrt{p_r}\boldsymbol{H}\boldsymbol{F}_{RF}\boldsymbol{F}_{BB}.\label{eq:H_bar}
\end{equation} 
The resulting estimate
\begin{equation}
\hat{\boldsymbol{s}}=\bar{\boldsymbol{H}}^{*}\boldsymbol{W}_{RF}\left[\boldsymbol{W}_{RF}{*}\left(\bar{\boldsymbol{H}}\bar{\boldsymbol{H}}^{*}+\boldsymbol{R}_{z}\right)\boldsymbol{W}_{RF}\right]^{-1}\boldsymbol{W}_{RF}^{*}\boldsymbol{y},
\end{equation}
and its MSE is equal to
\begin{equation}
\begin{aligned}
\mathbb{E}[\Vert&\boldsymbol{s}-\boldsymbol{\hat{s}}\Vert^2]=N_s-  \\  &\text{tr}\left(\bar{\boldsymbol{H}}^{*}\boldsymbol{W}_{RF}\left[\boldsymbol{W}_{RF}{*}\left(\bar{\boldsymbol{H}}\bar{\boldsymbol{H}}^{*}+\boldsymbol{R}_{z}\right)\boldsymbol{W}_{RF}\right]^{-1} \boldsymbol{W}_{RF}^{*}\bar{\boldsymbol{H}}\right).
\end{aligned}\label{eq:MSE}
\end{equation}

The digital combiner of (\ref{eq:Wbb}) is the optimal linear estimator in the MSE sense for any given precoders $\boldsymbol{F}_{RF},\boldsymbol{F}_{BB}$ and analog combiner $\boldsymbol{W}_{RF}$. Thus, our remaining goal is to design $\boldsymbol{W}_{RF},\boldsymbol{F}_{RF},\boldsymbol{F}_{BB}$ to minimize (\ref{eq:MSE}). This is equivalent to the following maximization problem:
\begin{equation}
\begin{aligned}
\underset{_{\boldsymbol{W}_{RF},\boldsymbol{F}_{RF},\boldsymbol{F}_{BB}}}{\max}&
f\left(\boldsymbol{W}_{RF},\boldsymbol{F}_{RF},\boldsymbol{F}_{BB}\right)&\\
\text{s.t.}:\;\;\;\;\;\;\;\;&
\boldsymbol{F}_{RF}\in\mathcal{F}^{N_{t}\times N_{RF}^{t}},\;\Vert\boldsymbol{F}_{RF}\boldsymbol{F}_{BB}\Vert_{F}^{2}\leq N_{s},& \\
&\boldsymbol{W}_{RF}\in\mathcal{W}^{N_{r}\times N_{RF}^{r}},&
\end{aligned}
\label{eq:Opt1}
\end{equation} 
with
\begin{align}
f\left(\boldsymbol{W}_{RF},\boldsymbol{F}_{RF},\boldsymbol{F}_{BB}\right)=&\\
\text{tr}\big(\boldsymbol{W}_{RF}^{*}\bar{\boldsymbol{H}}\bar{\boldsymbol{H}}^{*}\boldsymbol{W}_{RF}&\big[\boldsymbol{W}_{RF}{*}\left(\bar{\boldsymbol{H}}\bar{\boldsymbol{H}}^{*}+\boldsymbol{R}_{z}\right)\boldsymbol{W}_{RF}\big]^{-1} \big).\nonumber
\end{align}  

Similar to the case considered in \cite{ayach_spatially_2014}, the joint optimization problem (\ref{eq:Opt1}) is non-convex and has no known solution. Hence, as in previous works \cite{ayach_spatially_2014,yu_alternating_2016,yu_partially-connected_2017}, we simplify the problem by decoupling it. First, in Section \ref{sec:Precoder}, we optimize the precoders $\boldsymbol{F}_{RF},\boldsymbol{F}_{BB}$ given $\boldsymbol{W}_{RF}$. Then, in Section \ref{sec:Combiner}, we assume fixed $\boldsymbol{F}_{RF},\boldsymbol{F}_{BB}$ and optimize the combiner $\boldsymbol{W}_{RF}$. 

\section{Precoder design}\label{sec:Precoder}

We now consider the problem of designing the precoder $\boldsymbol{F}=\boldsymbol{F}_{RF}\boldsymbol{F}_{BB}$. 
For this purpose we assume a fixed analog combiner and rewrite (\ref{eq:Opt1}) as  
\begin{equation}
\begin{aligned}
\underset{_{\boldsymbol{F}_{RF},\boldsymbol{F}_{BB}}}{\max}&
\text{tr}\left(\tilde{\boldsymbol{H}}\boldsymbol{F}\boldsymbol{F}^{*}\tilde{\boldsymbol{H}}^{*}
\left[\tilde{\boldsymbol{H}}\boldsymbol{F}\boldsymbol{F}^{*}\tilde{\boldsymbol{H}}^{*}+\tilde{\boldsymbol{R}}\right]^{-1} \right)\\
\text{s.t.}:\;\;\;\;&
\boldsymbol{F}_{RF}\in\mathcal{F}^{N_{t}\times N_{RF}^{t}},\;\Vert\boldsymbol{F}_{RF}\boldsymbol{F}_{BB}\Vert_{F}^{2}\leq N_{s},
\\\end{aligned}
\label{eq:OptF}
\end{equation}
with $\tilde{\boldsymbol{H}}=\sqrt{p_r}\boldsymbol{W}_{RF}^{*}\boldsymbol{H}$ and $\tilde{\boldsymbol{R}}=\boldsymbol{W}_{RF}^{*}\boldsymbol{R}_{z}\boldsymbol{W}_{RF}$. If we relax the constraint $\boldsymbol{F}_{RF}\in\mathcal{F}^{N_{t}\times N_{RF}^{t}}$ to the fully-digital case $\boldsymbol{F}_{RF}\in\mathbb{C}^{N_t\times N_{RF}^t}$, then (\ref{eq:OptF}) has a closed form solution \cite{scaglione_optimal_2002}
\begin{equation}
\boldsymbol{F}_{opt} = \boldsymbol{V}\boldsymbol{\Phi}\boldsymbol{T}\label{eq:Fopt}
\end{equation}
where $\boldsymbol{V}$ contains the first $N_{RF}^t$ eigenvectors of the matrix $\tilde{\boldsymbol{H}}^{*}\tilde{\boldsymbol{R}}^{-1}\tilde{\boldsymbol{H}}$, $\boldsymbol{\Phi}$ is an $N_{RF}^t\times N_{RF}^t$ diagonal matrix with power allocation weights as in \cite{scaglione_optimal_2002}, and $\boldsymbol{T}$ denotes any $N_{RF}^t\times N_{RF}^t$ unitary matrix. It follows that there is a set of optimal solutions, one for every unitary $\boldsymbol{T}$. 

Previous works \cite{alkhateeb_hybrid_2013,ayach_spatially_2014,lee_hybrid_2015,mendez-rial_hybrid_2016,yu_alternating_2016,yu_partially-connected_2017} considered $\boldsymbol{T}=\boldsymbol{I}_{N_{RF}^t}$ and suggested to approximate $\boldsymbol{F}_{opt}$ with a decomposition $\boldsymbol{F}_{RF}\boldsymbol{F}_{BB}$, which yields the problem 

\begin{equation}
\begin{aligned}
\underset{_{\boldsymbol{F}_{RF},\boldsymbol{F}_{BB}}}{\min} & \Vert\boldsymbol{F}_{opt}-\boldsymbol{F}_{RF}\boldsymbol{F}_{BB}\Vert_{F}^{2} \\ 
\text{s.t.}:\;\;\;\; & \boldsymbol{F}_{RF}\in\mathcal{F}^{N_{t}\times N_{RF}^{t}},\;\Vert\boldsymbol{F}_{RF}\boldsymbol{F}_{BB}\Vert_{F}^{2}\leq N_{s}. 
\end{aligned}\label{eq:approxF}
\end{equation} 
This approach is motivated by the approximation in \cite{ayach_spatially_2014} which shows that minimizing (\ref{eq:approxF}) approximately maximizes the system's spectral efficiency. However, these solutions suffer from several disadvantages. 
In \cite{alkhateeb_hybrid_2013,ayach_spatially_2014,kim_mse-based_2015,lee_hybrid_2015,mendez-rial_hybrid_2016}, the suggested design for the analog precoder is based on a dictionary composed of candidate vectors. This constraint leads to a large performance gap from the fully-digital precoder, especially in the case where the number of RF chains is much less than the number of multipath components in the channel $\boldsymbol{H}$. 
In \cite{yu_alternating_2016}, the author suggests a manifold optimization algorithm termed MO-AltMin that achieves good performance, but has very high complexity and is suitable only for scheme S\ref{scheme:FullyPhaseOnly}. 
The solution in \cite{yu_partially-connected_2017} is based on the network S\ref{scheme:PhaseSubArrays} but with twice the amount of phase shifters per RF chains, and is irrelevant for other schemes. None of the above works take advantage of the flexibility in choosing $\boldsymbol{T}$, which we show below can boost the performance of existing solutions with only a small computational price.

Specifically, we exploit the flexibility in the choice of $\boldsymbol{T}$, and find the unitary matrix that results in the smallest approximation gap $\Vert\boldsymbol{V}\boldsymbol{\Phi}\boldsymbol{T}-\boldsymbol{F}_{RF}\boldsymbol{F}_{BB}\Vert_{F}^{2}$. To this end, we develop an alternating minimization framework inspired by \cite{Yu:2003:MSC:946247.946658}, that alternately optimizes over $\boldsymbol{F}_{RF},\boldsymbol{F}_{BB}$ and $\boldsymbol{T}$. In the simulations in Section \ref{sec:simulations}, we demonstrate that by adding the optimization over $\boldsymbol{T}$, we are able to improve the performance of state-of-the-art algorithms, with no substantial increase in complexity. 

\subsection{Alt-MaG - Alternating Minimization of Approximation Gap}

Plugging (\ref{eq:Fopt}) into (\ref{eq:approxF}) and adding $\boldsymbol{T}$ as an optimization variable, our problem becomes
\begin{equation}
\begin{aligned}
\underset{_{\boldsymbol{T},\boldsymbol{F}_{RF},\boldsymbol{F}_{BB}}}{\min} & \Vert\boldsymbol{V}\boldsymbol{\Phi}\boldsymbol{T}-\boldsymbol{F}_{RF}\boldsymbol{F}_{BB}\Vert_{F}^{2} \\ 
\text{s.t.}:\;\;\;\;\; & \boldsymbol{F}_{RF}\in\mathcal{F}^{N_{t}\times N_{RF}^{t}},\;\Vert\boldsymbol{F}_{RF}\boldsymbol{F}_{BB}\Vert_{F}^{2}\leq N_{s},\\
& \boldsymbol{T}\in\mathcal{U}^{N_{RF}^t}, 
\end{aligned}\label{eq:OptTF}
\end{equation} 
with $\mathcal{U}^{M}$ the set of all $M\times M$ unitary matrices.
To approximate the solution of (\ref{eq:OptTF}), we consider an alternating minimization approach. 
First, we fix $\boldsymbol{F}_{RF},\boldsymbol{F}_{BB}$ and solve 

\begin{equation}
\begin{aligned}
\underset{_{\boldsymbol{T}}}{\min\;} \;\;& \;\;\Vert\boldsymbol{V}\boldsymbol{\Phi}\boldsymbol{T}-\boldsymbol{F}_{RF}\boldsymbol{F}_{BB}\Vert_{F}^{2}&&& \\ 
\text{s.t.}:\;\;\; &\;\; \boldsymbol{T}\in\mathcal{U}^{N_{RF}^t}.&&&
\end{aligned}\label{eq:OptT}
\end{equation}
Next, using the resulting $\boldsymbol{T}$, we solve

\begin{equation}
\begin{aligned}
\underset{_{\boldsymbol{F}_{RF},\boldsymbol{F}_{BB}}}{\min\;} & \Vert\boldsymbol{V}\boldsymbol{\Phi}\boldsymbol{T}-\boldsymbol{F}_{RF}\boldsymbol{F}_{BB}\Vert_{F}^{2} \\ 
\text{s.t.}:\;\;\;\; & \boldsymbol{F}_{RF}\in\mathcal{F}^{N_{t}\times N_{RF}^{t}},\Vert\boldsymbol{F}_{RF}\boldsymbol{F}_{BB}\Vert_{F}^{2}\leq N_{s},
\end{aligned}\label{eq:OptFF}
\end{equation} 
and repeat iteratively. 

If both (\ref{eq:OptT}) and (\ref{eq:OptFF}) can be solved, then the objective value decreases in each step, and since it is bounded from below, the algorithm will converge to  a local optimum. However, while (\ref{eq:OptT}) has a (closed form) optimal solution, problem (\ref{eq:OptFF}) is a special case of (\ref{eq:approxF}), which has no known solution. To attain a sub optimal solution, any of the mentioned algorithms \cite{alkhateeb_hybrid_2013,ayach_spatially_2014,kim_mse-based_2015,lee_hybrid_2015,mendez-rial_hybrid_2016,yu_alternating_2016,yu_partially-connected_2017} can be used (provided it is compatible with the hardware choice). Each of these solutions will result in a different algorithm from the Alt-MaG family.

A closed form solution to problem (\ref{eq:OptT}) is given in the next theorem.
\begin{thm}\label{eq:Ttheorem}
	Given the svd decomposition $\boldsymbol{F}_{BB}^{*}\boldsymbol{F}_{RF}^{*}\boldsymbol{V}\boldsymbol{\Phi}=\tilde{\boldsymbol{U}}\boldsymbol{\Lambda}\tilde{\boldsymbol{V}}^{*}$, the optimal solution $\boldsymbol{T}_{opt}$ to (\ref{eq:OptT}) is given by
	\begin{equation}
	\boldsymbol{T}_{opt}=\tilde{\boldsymbol{V}}\tilde{\boldsymbol{U}}^{*}
	\end{equation}
\end{thm}
\begin{proof}
	First note that $\Vert\boldsymbol{V}\boldsymbol{\Phi}\boldsymbol{T}-\boldsymbol{F}_{RF}\boldsymbol{F}_{BB}\Vert_{F}^{2}=\Vert\boldsymbol{V}\boldsymbol{\Phi}\Vert_{F}^{2}+\Vert\boldsymbol{F}_{RF}\boldsymbol{F}_{BB}\Vert_{F}^{2}-2 \Re(\text{tr}\left(\boldsymbol{F}_{BB}^{*}\boldsymbol{F}_{RF}^{*}\boldsymbol{V}\boldsymbol{\Phi}\boldsymbol{T}\right))$. Thus, it is sufficient to maximize $\Re(\text{tr}\left(\boldsymbol{F}_{BB}^{*}\boldsymbol{F}_{RF}^{*}\boldsymbol{V}\boldsymbol{\Phi}\boldsymbol{T}\right))$.
	Denote $\boldsymbol{\Lambda}=\boldsymbol{\Omega}^{2}$. Then,
	\begin{align*}    
	%    \begin{aligned}    
	\Re(\text{tr}(&\boldsymbol{F}_{BB}^{*}\boldsymbol{F}_{RF}^{*}\boldsymbol{V}\boldsymbol{\Phi}\boldsymbol{T})) \\
	&\leq
	\left|\text{tr}(\boldsymbol{F}_{BB}^{*}\boldsymbol{F}_{RF}^{*}\boldsymbol{V}\boldsymbol{\Phi}\boldsymbol{T})\right|\\
	&=\left|\text{tr}((\tilde{\boldsymbol{U}}\boldsymbol{\Omega}\right)\left(\boldsymbol{T}^{*}\tilde{\boldsymbol{V}}\boldsymbol{\Omega})^{*})\right|
	\\
	&\underset{\left(*\right)}{\leq}\sqrt{\text{tr}(\tilde{\boldsymbol{U}}\boldsymbol{\Omega}^2\tilde{\boldsymbol{U}}^{*})\text{tr}(\boldsymbol{T}^{*}\tilde{\boldsymbol{V}}\boldsymbol{\Omega}^2\tilde{\boldsymbol{V}}^{*}\boldsymbol{T})}\\
	&=\text{tr}\left(\boldsymbol{\Omega}^2\right),
	\end{align*}
	where $\left(*\right)$ stems from the Cauchy-Schwarz inequality. For $\boldsymbol{T}=\tilde{\boldsymbol{V}}\tilde{\boldsymbol{U}}^{*}$, the upper bound is achieved with equality.
\end{proof}

The resulting family of Alt-MaG algorithms is outlined in Algorithm \ref{algo:IterQuan}. 

\begin{algorithm}[h]
	\caption{Alt-MaG - alternating minimization of approximation gap}
	\text{\textbf{Input:} fully-digital precoder $\boldsymbol{V}\boldsymbol{\Phi}$, threshold $t$}    
	\text{\textbf{Output:} digital precoder $\boldsymbol{F}_{BB}$, analog precoder $\boldsymbol{F}_{RF}$}
	\text{\textbf{Initialize} $\boldsymbol{T}=\boldsymbol{I}_{N_{RF}^t},\boldsymbol{F}_{RF}=\boldsymbol{0},\boldsymbol{F}_{BB}=\boldsymbol{0}$}
	\text{\textbf{While} $\Vert\boldsymbol{V}\boldsymbol{\Phi}\boldsymbol{T}-\boldsymbol{F}_{RF}\boldsymbol{F}_{BB}\Vert_{F}^{2}\geq t$ do:}
	\begin{enumerate}
		\item \text{\textbf{Calculate} $\left(\boldsymbol{F}_{RF},\boldsymbol{F}_{BB}\right)$ from (\ref{eq:OptFF}) using one of} \\ \text{the methods in Section \ref{sec:Precoder}}
		\item \text{\textbf{Calculate} the svd decomposition $\boldsymbol{F}_{BB}^{*}\boldsymbol{F}_{RF}^{*}\boldsymbol{V}\boldsymbol{\Phi}=\tilde{\boldsymbol{U}}\boldsymbol{\Lambda}\tilde{\boldsymbol{V}}^{*}$}
		\item \text{\textbf{Set} $\boldsymbol{T}=\tilde{\boldsymbol{V}}\tilde{\boldsymbol{U}}^{*}$}
	\end{enumerate}\label{algo:IterQuan}
\end{algorithm} 

Next we suggest a possible simplification to (\ref{eq:OptFF}) that can be used with any of the schemes S\ref{scheme:FullyFirst}-S\ref{scheme:FlexiLast}, and results in a simple closed form solution for $\boldsymbol{F}_{RF},\boldsymbol{F}_{BB}$. 

\subsection{MaGiQ - Minimal Gap Iterative Quantization}\label{sec:Methods}

The complexity of Algorithm \ref{algo:IterQuan} is determined by the method used in step 1 and could be accordingly high. One approach that yields a low complexity algorithm, is approximating $\boldsymbol{V}\boldsymbol{\Phi}\boldsymbol{T}$ by an analog precoder only, setting the digital precoder to $\boldsymbol{F}_{BB}=\boldsymbol{I}_{N_{RF}^{t}}$. 
That is, in every iteration, solving

\begin{equation}
\begin{aligned}
\underset{_{\boldsymbol{F}_{RF}}}{\min\;\;\;} & \;\Vert\boldsymbol{V}\boldsymbol{\Phi}\boldsymbol{T}-\boldsymbol{F}_{RF}\Vert_{F}^{2} \\ 
\text{s.t.}:\;\;\; & \;\boldsymbol{F}_{RF}\in\mathcal{F}^{N_{t}\times N_{RF}^{t}}. \\
\end{aligned}\label{eq:OptFrf}
\end{equation} 
The benefit of this approach is the existence of a simple, optimal, closed form solution to (\ref{eq:OptFrf}), given by
\begin{equation}
\boldsymbol{F}_{RF}=\boldsymbol{P}_\mathcal{F}\left(\boldsymbol{V}\boldsymbol{\Phi}\boldsymbol{T}\right),
\end{equation}
with $\boldsymbol{P}_\mathcal{F}\left(\cdot\right)$ defined in Section \ref{sec:SignalModel}.  

Once $\boldsymbol{F}_{RF}$ and $\boldsymbol{T}$ converge, $\boldsymbol{F}_{BB}$ is calculated as the solution to
\begin{equation}
\begin{aligned}
\underset{_{\boldsymbol{F}_{BB}}}{\min\;\;\;} & \;\Vert\boldsymbol{V}\boldsymbol{\Phi}\boldsymbol{T}-\boldsymbol{F}_{RF}\boldsymbol{F}_{BB}\Vert_{F}^{2} \\ 
\text{s.t.}:\;\;\; & \;\Vert\boldsymbol{F}_{RF}\boldsymbol{F}_{BB}\Vert_{F}^{2}\leq N_{S}. \\
\end{aligned}\label{eq:OptFbb}
\end{equation}
Note that $\boldsymbol{F}_{RF}$ depends on $\boldsymbol{V}\boldsymbol{\Phi}\boldsymbol{T}$ which are a function of the random channel and interference. With high probability, it has full column rank. In this case, (\ref{eq:OptFbb}) has a closed form solution
\begin{equation}
\boldsymbol{F}_{BB}=\left[\boldsymbol{F}_{RF}^{*}\boldsymbol{F}_{RF}\right]^{-1}\boldsymbol{F}_{RF}^{*}\boldsymbol{V}\boldsymbol{\Phi}\boldsymbol{T}.\label{eq:Fbb}
\end{equation}
Note that since the optimal solution $\boldsymbol{V}\boldsymbol{\Phi}\boldsymbol{T}$ obeys the transmit power constraint $\Vert\boldsymbol{V}\boldsymbol{\Phi}\boldsymbol{T}\Vert_{F}^{2}\leq N_S$, it follows that
\begin{flalign}
\Vert\boldsymbol{F}_{RF}\boldsymbol{F}_{BB}\Vert_{F}^{2}&=\Vert\boldsymbol{F}_{RF}\left[\boldsymbol{F}_{RF}^{*}\boldsymbol{F}_{RF}\right]^{-1}\boldsymbol{F}_{RF}^{*}\boldsymbol{V}\boldsymbol{\Phi}\boldsymbol{T}\Vert_{F}^{2}&\nonumber\\
&=\Vert\boldsymbol{P}_{\boldsymbol{F}_{RF}}\boldsymbol{V}\boldsymbol{\Phi}\boldsymbol{T}\Vert_{F}^{2}\leq \Vert\boldsymbol{V}\boldsymbol{\Phi}\boldsymbol{T}\Vert_{F}^{2}\leq N_{S}.&\nonumber
\end{flalign}
Therefore, (\ref{eq:Fbb}) complies with the transmit power constraint.

MaGiQ is summarized in Algorithm~\ref{algo:SimpleQuan}.
Since the non-negative objective function $\Vert\boldsymbol{V}\boldsymbol{\Phi}\boldsymbol{T}-\boldsymbol{F}_{RF}\Vert_{F}^{2}$ does not increase at each iteration of Algorithm \ref{algo:SimpleQuan}, the algorithm converges.
In the simulations in Section~\ref{sec:simulations}, we demonstrate that MaGiQ has lower MSE than other low-complexity methods in the case of very few RF chains.
\begin{algorithm}[h] 
	\caption{MaGiQ -  Minimal Gap Iterative Quantization}
	\text{\textbf{Input:} fully-digital precoder $\boldsymbol{V}\boldsymbol{\Phi}$, threshold $t$}    
	\text{\textbf{Output:} digital precoder $\boldsymbol{F}_{BB}$, analog precoder $\boldsymbol{F}_{RF}$}
	\text{\textbf{Initialize} $\boldsymbol{T}=\boldsymbol{I}_{N_{RF}^t},\boldsymbol{F}_{RF}=\boldsymbol{0},\boldsymbol{F}_{BB}=\boldsymbol{I}_{N_{RF}^t}$}
	\text{\textbf{While} $\Vert\boldsymbol{V}\boldsymbol{\Phi}\boldsymbol{T}-\boldsymbol{F}_{RF}\Vert_{F}^{2}\geq t$ do:}
	\begin{enumerate}
		\item \textbf{Calculate }$\boldsymbol{F}_{RF}=\boldsymbol{P}_{\mathcal{F}}\left(\boldsymbol{V}\boldsymbol{\Phi}\boldsymbol{T}\right)$\label{eq:Fstep}
		\item \text{\textbf{Calculate} the svd decomposition $\boldsymbol{F}_{RF}^{*}\boldsymbol{V}\boldsymbol{\Phi}=\tilde{\boldsymbol{U}}\boldsymbol{\Lambda}\tilde{\boldsymbol{V}}^{*}$}
		\item \text{\textbf{Set} $\boldsymbol{T}=\tilde{\boldsymbol{V}}\tilde{\boldsymbol{U}}^{*}$}\label{eq:Tstep}
	\end{enumerate}
	\textbf{Calculate } $\boldsymbol{F}_{BB}=\left[\boldsymbol{F}_{RF}^{*}\boldsymbol{F}_{RF}\right]^{-1}\boldsymbol{F}_{RF}^{*}\boldsymbol{V}\boldsymbol{\Phi}\boldsymbol{T}$
	\label{algo:SimpleQuan}
\end{algorithm}

A version of Algorithm~\ref{algo:SimpleQuan}, PE-AltMin, has been previously suggested in \cite{yu_alternating_2016} for the scheme S\ref{scheme:FullyPhaseOnly}, as an optimization of an upper bound on (\ref{eq:OptFF}). There, the authors assumed that the optimal $\boldsymbol{F}_{BB}$ is an orthogonal matrix of the form $\boldsymbol{F}_{BB}=\alpha\boldsymbol{T}$, and considered the problem
\begin{equation}
\begin{aligned}
\underset{_{\boldsymbol{T},\alpha,\boldsymbol{F}_{RF}}}{\min} & \Vert\boldsymbol{V}\boldsymbol{\Phi}-\alpha\boldsymbol{F}_{RF}\boldsymbol{T}\Vert_{F}^{2} \\ 
\text{s.t.}:\;\;\; & \boldsymbol{F}_{RF}\in\mathcal{F}^{N_{t}\times N_{RF}^{t}},\;\Vert\alpha\boldsymbol{F}_{RF}\boldsymbol{T}\Vert_{F}^{2}=N_{s},\\
& \boldsymbol{T}\in\mathcal{U}^{N_{RF}^t}. 
\end{aligned}\label{eq:OptAltMin}
\end{equation} 
Setting $\alpha=\frac{\sqrt{N_{S}}}{\Vert\boldsymbol{F}_{RF}\boldsymbol{T}\Vert_{F}}$, and using some mathematical manipulations, they bounded the objective in (\ref{eq:OptAltMin}) from above with the expression $\Vert\boldsymbol{V}\boldsymbol{\Phi}\boldsymbol{T}^*-\boldsymbol{F}_{RF}\Vert_{F}^{2}$. After optimizing the upper bound over $\boldsymbol{T}$ and $\boldsymbol{F}_{RF}$ alternately, they set $\boldsymbol{F}_{BB}$ as $\boldsymbol{F}_{BB}=\frac{\sqrt{N_{S}}}{\Vert\boldsymbol{F}_{RF}\boldsymbol{T}\Vert_{F}}\boldsymbol{T}$, which is sub optimal, as demonstrated in Section \ref{sec:simulations}, where we compare both methods. 
The main difference between this approach and MaGiQ, is that while we exploit the degree of freedom in the fully-digital precoder, PE-AltMin restricts the set of possible digital precoders to orthogonal matrices alone, which degrades performance. 
In addition, Algorithm \ref{algo:SimpleQuan} is relevant to any hardware structure.

Alt-MaG and MaGiQ where developed for MSE minimization, but can be generalized to other performance criteria, by replacing the optimal unconstrained solution expressions. For example, when considering the system's spectral efficiency as in \cite{ayach_spatially_2014,yu_alternating_2016}, the optimal solution is (\ref{eq:Fopt}) with $\boldsymbol{\Phi}$ now being a diagonal matrix with the power allocation weights given by water filling \cite{scaglione_optimal_2002}.

\subsection{Optimality of MaGiQ}\label{sec:Optimality}
One advantage of MaGiQ over other existing techniques is that in some special cases it converges to the optimal fully-digital solution, as we show in the following proposition.
\begin{prop}\label{prop:MaGiQ}
	Define $\boldsymbol{F}_{opt}$ as in (\ref{eq:Fopt}), and assume that $\boldsymbol{\Phi}$ has positive values on its diagonal. Denote  $\boldsymbol{\tilde{F}}$ the hybrid precoder produced by MAGiQ for the input $\boldsymbol{F}_{opt}$. If $\boldsymbol{V}\boldsymbol{\Psi}\in\mathcal{F}^{N_{t}\times N_{RF}^{t}}$, for some diagonal, positive definite matrix $\boldsymbol{\Psi}$, then 
	\begin{equation*}
	\boldsymbol{\tilde{F}}=\boldsymbol{V}\boldsymbol{\Phi}.
	\end{equation*} 
\end{prop}
Proposition~\ref{prop:MaGiQ} implies that if each column of $\boldsymbol{V}$ lies in the feasible set $\mathcal{F}^{N_{t}\times 1}$ up to some scaling, then MaGiQ will produce a globally optimal solution. This property is not assured in other algorithms, as demonstrated in the simulations in Section~\ref{sec:simulations}.
\begin{proof}
	For simplicity, we consider here the fully-connected phase shifter network S\ref{scheme:FullyPhaseOnly}. However, Proposition~\ref{prop:MaGiQ} holds for any hardware scheme with the same condition $\boldsymbol{V}\boldsymbol{\Psi}\in\mathcal{F}^{N_{t}\times N_{RF}^{t}}$, and similar proofs can be constructed accordingly.
	
	For S\ref{scheme:FullyPhaseOnly}, $\boldsymbol{V}\boldsymbol{\Psi}\in\mathcal{F}^{N_{t}\times N_{RF}^{t}}$ implies that $\boldsymbol{V}$ has columns with constant modulus.
	In the first iteration of MaGiQ, $\boldsymbol{T}$ is initialized as $\boldsymbol{T}=\boldsymbol{I}_{N_{RF}^t}$. Since $\boldsymbol{\Phi}$ is a diagonal matrix with positive diagonal values, step~1 of the algorithm leads to 
	\begin{equation*}
	\boldsymbol{F}_{RF}=\boldsymbol{P}_{\mathcal{F}}\left(\boldsymbol{V}\boldsymbol{\Phi}\right)=\boldsymbol{V}\boldsymbol{\Psi}.
	\end{equation*}
	It follows that $\boldsymbol{F}_{RF}^{*}\boldsymbol{V}\boldsymbol{\Phi}=\boldsymbol{\Psi}\boldsymbol{\Phi}$, which is diagonal. 
	Thus, step~3 yields 
	\begin{equation*}
	\boldsymbol{T}=\boldsymbol{I}_{N_{RF}^t}ץ
	\end{equation*}    
	and the algorithm converges after a single iteration. 
	For $\boldsymbol{F}_{BB}$ we get
	\begin{equation*}
	\boldsymbol{F}_{BB}=\left[\boldsymbol{F}_{RF}^{*}\boldsymbol{F}_{RF}\right]^{-1}\boldsymbol{F}_{RF}^{*}\boldsymbol{V}\boldsymbol{\Phi}=\boldsymbol{\Psi}^{-1}\boldsymbol{\Phi},
	\end{equation*}
	and the received hybrid precoder is $\boldsymbol{\tilde{F}}=\boldsymbol{V}\boldsymbol{\Phi}$.
\end{proof}

Note that the diagonal matrix $\boldsymbol{\Phi}$ represents non-negative weights \cite{scaglione_optimal_2002}. If one of the weights is equal to zero, then the corresponding column in the optimal fully-digital precoder is all zeros. In that case, we can drop the relevant column, and use the above proposition with the remaining fully-digital precoder. 

One example of a scenario in which the conditions of Proposition \ref{prop:MaGiQ} hold, is a fully-connected phase shifter precoder (for simplicity we assume fully-digital combiner) with $\boldsymbol{R}_{z}=\sigma_{z}^2\boldsymbol{I}_{N_r}$, and circulant channel matrix, that is 
\begin{equation}
\boldsymbol{H}=\boldsymbol{A}_{r}\boldsymbol{\Lambda}\boldsymbol{A}_{t}^{*},\label{eq:VirtualModel}
\end{equation}
where $\boldsymbol{A}_{r},\boldsymbol{A}_{t}$ are DFT matrices and $\boldsymbol{\Lambda}$ is a diagonal gain matrix. 
In this case, the singular vectors of $\boldsymbol{H}^{*}\boldsymbol{R}_{z}^{-1}\boldsymbol{H}$ are columns of the DFT matrix which are unimodular, and thus lie in $\mathcal{F}^{N_{t}\times 1}$ for S\ref{scheme:FullyPhaseOnly}.  Therefore, from Proposition \ref{prop:MaGiQ}, MaGiQ produces the optimal solution (\ref{eq:Fopt}).

A common special case of (\ref{eq:VirtualModel}), is the well-known narrow-band mmWave clustered channel \cite{sayeed_deconstructing_2002,smulders_characterisation_1997,xu_spatial_2002}, with a single ray at each cluster and asymptotic number of antennas. 
This model consists of $N_{cl}$ clusters with $N_{ray}$ propagating rays, each ray associated with transmit and receive directions, and complex gain. The channel matrix can be written as 
\begin{equation}
\boldsymbol{H}=\sqrt{\frac{N_{t}N_{r}}{N_{cl}N_{ray}}}\sum_{i=1}^{N_{cl}}\sum_{l=1}^{N_{ray}}\alpha_{il}\boldsymbol{a}_{r}\left(\phi_{il}^{r},\theta_{il}^{r}\right)\boldsymbol{a}_{t}\left(\phi_{il}^{t},\theta_{il}^{t}\right)^{*},\label{eq:mmWaveChannel}
\end{equation}
with $\alpha_{il}$ the complex gain of the $l$th ray in the $i$th cluster, and $\left(\phi_{il}^{t},\theta_{il}^{t}\right)$,  $\left(\phi_{il}^{r},\theta_{il}^{r}\right)$ the azimuth and elevation angles of departure and arrival for the $l$th ray in the $i$th cluster respectively. 
The vectors $\boldsymbol{a}_{t}\left(\phi_{il}^{t},\theta_{il}^{t}\right)$ and $\boldsymbol{a}_{r}\left(\phi_{il}^{r}\theta_{il}^{r}\right)$ represent the transmit and receive array responses, and depend on the array geometry.

For $N_{ray}=1$, (\ref{eq:mmWaveChannel}) can be written as (\ref{eq:VirtualModel}) with $\boldsymbol{\Lambda}$ a $N_{cl}\times N_{cl}$ diagonal matrix with the elements $\sqrt{\frac{N_{t}N_{r}}{N_{cl}N_{ray}}}\alpha_{il}$ on its diagonal, and $\boldsymbol{A}_{r},\boldsymbol{A}_{t}$ being steering matrices composed of the vectors $\boldsymbol{a}_{r}\left(\phi_{il}^{r},\theta_{il}^{r}\right)$ and $\boldsymbol{a}_{t}\left(\phi_{il}^{t},\theta_{il}^{t}\right)$ respectively. 
As $N_{r},N_{t}$ grow to infinity, the matrices $\boldsymbol{A}_{r},\boldsymbol{A}_{t}$ become orthogonal. 
Hence, in a noise-limited system, $\boldsymbol{V}$ is asymptotically equal to the steering vectors (up to some normalization constant) and thus unimodular. 

The above example can be generalized easily to a hybrid fully-connected phase shifter combiner rather than a fully-digital one. 

To summarize the results of this section, MaGiQ is a low complexity algorithm that yields low MSE with very few RF chains. When the number of chains increases, Alt-Mag can be used with other methods to solve problem (\ref{eq:OptFF}) in its first step. The choice of algorithm should be made according to the complexity requirement. For example, if the computational load is not a limiting factor, then MO-AltMin from \cite{yu_alternating_2016} is a good choice. 

\section{Analog Combiner Design}\label{sec:Combiner}
We now turn to design the analog combiner 
\begin{equation}
\boldsymbol{W}=\boldsymbol{W}_{BB}\boldsymbol{W}_{RF}.\label{eq:HybridDec}
\end{equation}
For this purpose we assume a fixed precoder, that is, $\bar{\boldsymbol{H}}$ in (\ref{eq:H_bar}) is fixed and known. 

In (\ref{eq:Wbb}), we derived the optimal $N_{RF}^r\times N_{RF}^r$ digital combiner, given the hybrid decomposition (\ref{eq:HybridDec}), and reduced the problem to $\boldsymbol{W}_{RF}$ alone. 
A different approach is to first allow the entire $N_r\times N_{RF}^r$ hybrid combiner to be digital, and then try and approximate the fully-digital solution with a hybrid decomposition, similar to the precoder optimization in the previous section.   

This approach was adopted in \cite{ayach_spatially_2014,kim_mse-based_2015}. 
There, the authors begun by allowing $\boldsymbol{W}$ to be fully-digital, that is $\boldsymbol{W}\in\mathbb{C}^{N_r\times N_{RF}^r}$, and solved problem (\ref{eq:Opt0}) for fixed precoders. This yields the MMSE estimator of $\boldsymbol{s}$ from $\boldsymbol{y}$, given by
\begin{equation}
\boldsymbol{W}_{mmse} = \boldsymbol{B}^{-1}\bar{\boldsymbol{H}}, \label{eq:Wmse}
\end{equation}
with
\begin{equation}
\boldsymbol{B}=\bar{\boldsymbol{H}}\bar{\boldsymbol{H}}^{*}+\boldsymbol{R}_{z}.\label{eq:B}
\end{equation}
It was shown in \cite{ayach_spatially_2014} that minimizing (\ref{eq:Opt0}) (for fixed precoders) is equivalent to solving the following optimization problem:   
\begin{equation}
\begin{aligned}
\underset{_{\boldsymbol{W}_{RF},\boldsymbol{W}_{BB}}}{\min} & 
\Vert \boldsymbol{B}^{\frac{1}{2}}\left(\boldsymbol{W}_{mmse}-\boldsymbol{W}_{RF}\boldsymbol{W}_{BB}\right)\Vert_{F}^{2}\\ 
\text{s.t.}:\;\;\;\;\; & \boldsymbol{W}_{RF}\in\mathcal{W}^{N_{r}\times N_{RF}^{r}}.
\end{aligned}\label{eq:HeathProb}
\end{equation} 
That is, finding the hybrid decomposition of $\boldsymbol{W}_{mmse}$ that minimizes the weighted norm in (\ref{eq:HeathProb}).

Unlike the precoder case, here (\ref{eq:HeathProb}) and (\ref{eq:Opt0}) are equivalent. Motivated by this, the authors in \cite{ayach_spatially_2014,kim_mse-based_2015} assumed a multipath channel structure as in (\ref{eq:mmWaveChannel}), which implies that the optimal combiner is a sparse sum of steering vectors. They then suggested a variation of simultaneous OMP (SOMP) to solve (\ref{eq:HeathProb}), that chooses the analog combiner vectors from a steering dictionary and sets the digital combiner as the corresponding weights.  However, this algorithm suffers from degraded performance due to the dictionary constraint. 

By dropping the weights matrix $\boldsymbol{B}^{\frac{1}{2}}$, (\ref{eq:HeathProb}) becomes 
\begin{equation}
\begin{aligned}
\underset{_{\boldsymbol{W}_{RF},\boldsymbol{W}_{BB}}}{\min} & \;\Vert\boldsymbol{W}_{mmse}-\boldsymbol{W}_{RF}\boldsymbol{W}_{BB}\Vert_{F}^{2} \\ 
\text{s.t.}:\;\;\;\;\; & \;\boldsymbol{W}_{RF}\in\mathcal{W}^{N_{t}\times N_{RF}^{t}}, 
\end{aligned}\label{eq:approxW}
\end{equation} 
which constitutes an upper bound on (\ref{eq:HeathProb}) (divided by the constant factor $\Vert\boldsymbol{B}^{\frac{1}{2}}\Vert_{F}^2$), and is identical (up to the power constraint) to the precoder problem (\ref{eq:approxF}) in the previous section. Hence the methods in 
\cite{alkhateeb_hybrid_2013,lee_hybrid_2015,mendez-rial_hybrid_2016,yu_alternating_2016,yu_partially-connected_2017} can be used to solve it. However, solving (\ref{eq:approxW}) is only optimizing an upper bound on the original objective, and is thus sub-optimal. 

\subsection{MaGiQ for Combiner Design}
As opposed to previous works, we impose the hybrid decomposition (\ref{eq:HybridDec}), and optimize (\ref{eq:Opt0}) over both $\boldsymbol{W}_{BB}$ and $\boldsymbol{W}_{RF}$, as in Section~\ref{sec:DataEst}. This leads to the problem of minimizing (\ref{eq:MSE}) over $\boldsymbol{W}_{RF}$, which is equivalent to
\begin{equation}
\begin{aligned}
\underset{_{\boldsymbol{W}_{RF}}}{\max\;\;\;}&\;
\text{tr}\left(\boldsymbol{W}_{RF}^{*}\boldsymbol{A}\boldsymbol{W}_{RF}
\left[\boldsymbol{W}_{RF}^{*}\boldsymbol{B}\boldsymbol{W}_{RF}\right]^{-1}\right)\\
\text{s.t.}:\;\;\;&
\;\boldsymbol{W}_{RF}\in\mathcal{W}^{N_{r}\times N_{RF}^{r}}, \\
\end{aligned}
\label{eq:WrfOpt}
\end{equation} 
with $\boldsymbol{A}=\bar{\boldsymbol{H}}\bar{\boldsymbol{H}}^{*}$, $\boldsymbol{B}$ as in (\ref{eq:B}), and $\mathcal{W}$ the feasible set of the given hardware scheme as described in Section \ref{sec:SignalModel}. Note that since $\boldsymbol{R}_{z}$ is positive definite, it follows that $\boldsymbol{B}$ is invertible.

\begin{prop}
	The optimal fully-digital solution to (\ref{eq:WrfOpt}), is given by 
	\begin{equation}
	\boldsymbol{W}_{opt} = \boldsymbol{B}^{-\frac{1}{2}}\boldsymbol{U}\boldsymbol{S}, \label{eq:optimalW}
	\end{equation}
	where $\boldsymbol{U}$ contains the first $N_{RF}^r$ eigenvectors of the matrix $\boldsymbol{B}^{-\frac{1}{2}}\boldsymbol{A}\boldsymbol{B}^{-\frac{1}{2}}$, with $\boldsymbol{A}=\bar{\boldsymbol{H}}\bar{\boldsymbol{H}}^{*}$, and $\bar{\boldsymbol{H}},\boldsymbol{B}$ are as in (\ref{eq:H_bar}), and (\ref{eq:B}). The matrix $\boldsymbol{S}$ is any $N_{RF}^r\times N_{RF}^r$ invertible matrix.
\end{prop}
\begin{proof}
	The fully-digital problem corresponding to (\ref{eq:WrfOpt}) is
	\begin{equation}
	\begin{aligned}
	&\underset{_{\boldsymbol{W}_{RF}}}{\max}\;\;\;
	\text{tr}\left(\boldsymbol{W}_{RF}^{*}\boldsymbol{A}\boldsymbol{W}_{RF}
	\left[\boldsymbol{W}_{RF}^{*}\boldsymbol{B}\boldsymbol{W}_{RF}\right]^{-1}\right).&\\
	\end{aligned}
	\label{eq:WrfFullyDig}
	\end{equation}  
	Let $\boldsymbol{Q}\in\mathbb{C}^{N_{r}\times N_{RF}^r}$ be defined by $\boldsymbol{Q}=\boldsymbol{B}^{\frac{1}{2}}\boldsymbol{W}_{RF}$. 
	Then the objective in (\ref{eq:WrfFullyDig}) equals to $\text{tr}(\boldsymbol{P}_{\boldsymbol{Q}}\boldsymbol{B}^{-\frac{1}{2}}\boldsymbol{A}\boldsymbol{B}^{-\frac{1}{2}})$.
	Let $\boldsymbol{V}$ be a $N_r\times N_{RF}^r$ matrix with orthogonal columns that span $\mathcal{R}(\boldsymbol{Q})$. 
	Then, $\boldsymbol{P}_{\boldsymbol{Q}}=\boldsymbol{V}\boldsymbol{V}^{*}$ and $\boldsymbol{Q}=\boldsymbol{V}\boldsymbol{S}$
	for some invertible $N_{RF}^r\times N_{RF}^r$ matrix $\boldsymbol{S}$.
	Thus, the objective becomes $\text{tr}(\boldsymbol{V}^{*}\boldsymbol{B}^{-\frac{1}{2}}\boldsymbol{A}\boldsymbol{B}^{-\frac{1}{2}}\boldsymbol{V})$ which is maximized by choosing $\boldsymbol{V}$ as the first $N_{RF}^r$ eigenvectors of $\boldsymbol{B}^{-\frac{1}{2}}\boldsymbol{A}\boldsymbol{B}^{-\frac{1}{2}}$, i.e. $\boldsymbol{V}=\boldsymbol{U}$.
	Therefore,
	\begin{equation*}
	\boldsymbol{Q}=\boldsymbol{U}\boldsymbol{S},
	\end{equation*}
	which yields the solution (\ref{eq:optimalW}).
\end{proof}
Similarly to the previous section, we can now approximate the optimal unconstrained solution (\ref{eq:optimalW}) with an analog combiner, while looking for an $\boldsymbol{S}$ that yields a minimal approximation gap between the two. 

In (\ref{eq:optimalW}), $\boldsymbol{S}$ is only required to be invertible, in contrast to the unitary constraint in (\ref{eq:Fopt}). Hence, limiting it to $\mathcal{U}^{N_{RF}^r}$ is  restrictive. However, this is a convenient way to enforce the invertibility constraint. The optimization problem we consider is therefore
\begin{equation}
\begin{aligned}
\underset{_{\boldsymbol{W}_{RF},\boldsymbol{S}}}{\min\;\;} & \;\Vert \boldsymbol{B}^{-\frac{1}{2}}\boldsymbol{U}\boldsymbol{S}-\boldsymbol{W}_{RF}\Vert_{F}^{2} \\ 
\;\;\;\text{s.t.}:\;\;\; & \;\boldsymbol{W}_{RF}\in\mathcal{W}^{N_{r}\times N_{RF}^{r}},\\
&\;\boldsymbol{S}\in\mathcal{U}^{N_{RF}^r},
\end{aligned}\label{eq:WSopt}
\end{equation}
which can be solved using the MaGiQ approach presented in Section \ref{sec:Precoder}. Note that here Alt-MaG and MaGiQ coincide, since in step 1 of Alt-MaG the optimization is over $\boldsymbol{W}_{RF}$ alone. Once the algorithm converges, $\boldsymbol{W}_{BB}$ is calculated as in (\ref{eq:Wbb}).

MaGiQ enjoys good performance with very few RF chains, but when the number of chains increases, other methods are preferable.  
One of the main factors that degrade MaGiQ, is that it only optimizes a bound \eqref{eq:approxW} on \eqref{eq:HeathProb}, rather than directly maximizing  \eqref{eq:HeathProb} (which is equivalent to (\ref{eq:WrfOpt})). 
Next, we exploit the special structure of the ratio-trace objective in (\ref{eq:WrfOpt}), and construct a greedy method for its direct maximization over $\boldsymbol{W}_{RF}$. The proposed GRTM algorithm solves a ratio-of-scalars problem at each step and achieves low MSE when the number of RF chains increases. 

Both GRTM and MaGiQ enjoy low complexity and each has its merits.  
For a small number of RF chains or a noise-limited case with a fully connected phase shifters network and channel that has unimodular singular vectors, one should choose MaGiQ. For a more general channel model and increasing number of RF chains, GRTM is preferable.  

\subsection{GRTM - Greedy Ratio Trace Maximization}\label{sec:GreedyAlg}
We now describe the GRTM approach for directly solving (\ref{eq:WrfOpt}). The concept of GRTM is to solve (\ref{eq:WrfOpt}) in a greedy manner, where in each iteration we add one RF chain and choose the optimal combiner vector to add to the previously $K$ selected vectors, $1\leq K\leq N_{RF}^{r}$. 

Assume we have a solution $\boldsymbol{W}_{RF}^{\left(K\right)}\in\mathcal{W}^{N_r\times K}$ for the $K$-sized problem, i.e. (\ref{eq:WrfOpt}) with $N_{RF}^r=K$. We now want to add an additional RF chain, and compute the optimal column $\boldsymbol{w}$  such that $\boldsymbol{W}_{RF}^{\left(K+1\right)}=[\boldsymbol{W}_{RF}^{\left(K\right)}\;\boldsymbol{w}]$. 
First, note that in order for $\boldsymbol{W}_{RF}^{\left(K+1\right)\;*}\boldsymbol{B}\boldsymbol{W}_{RF}^{\left(K+1\right)}$ to be invertible and the objective in (\ref{eq:WrfOpt}) to be well defined for the $\left(K+1\right)$-sized case, we require $\boldsymbol{w}\notin\mathcal{R}(\boldsymbol{W}_{RF}^{\left(K\right)})$. In practice, this condition implies that each RF chain contributes new and independent information with respect to  the other chains.
Given this condition, the $\left(K+1\right)$-sized optimization problem is 
\begin{equation}
\begin{aligned}
\underset{\boldsymbol{w}}{\max}\;\;\;&\;
\text{tr}\left(\boldsymbol{W}_{RF}^{\left(K+1\right)\;*}\boldsymbol{A}\boldsymbol{W}_{RF}^{\left(K+1\right)}
\left[\boldsymbol{W}_{RF}^{\left(K+1\right)\;*}\boldsymbol{B}\boldsymbol{W}_{RF}^{\left(K+1\right)}\right]^{-1}\right)\\
\text{s.t.}:\;\;\;\;&\;
\boldsymbol{w}\in\mathcal{W}^{N_r\times 1},\;\boldsymbol{w}\notin\mathcal{R}\left(\boldsymbol{W}_{RF}^{\left(K\right)}\right).
\end{aligned}
\label{eq:KOpt}
\end{equation} 

In the next proposition, we show that (\ref{eq:KOpt}) is equivalent to solving the following vector optimization problem, referred to as the base case: 
\begin{equation}
\begin{aligned}
\underset{\boldsymbol{w}\;\;}{\max\;\;}\;\;\;&
\frac{\boldsymbol{w}^{*}\boldsymbol{C}\boldsymbol{w}}
{\boldsymbol{w}^{*}\boldsymbol{D}\boldsymbol{w}}&\\
\text{s.t.}:\;\;\;\;\;\;
&\boldsymbol{w}\in\mathcal{W}^{N_r\times 1},\;\boldsymbol{w}^{*}\boldsymbol{D}\boldsymbol{w}>0, \\
\end{aligned}
\label{eq:VectorProb}
\end{equation} 
where $\boldsymbol{D}$ and $\boldsymbol{C}$ are specific $N_r\times N_r$ positive semi definite (psd) matrices.

Using Proposition \ref{prop:GRTM} leads to the GRTM solution of (\ref{eq:WrfOpt}), where in each iteration we choose the best column vector to add to the previously selected combiner vectors by solving (\ref{eq:VectorProb}). The GRTM algorithm is summarized in Algorithm \ref{algo:Greedy}.
\begin{prop}\label{prop:GRTM}
	Problems (\ref{eq:VectorProb}) and (\ref{eq:KOpt}) are equivalent with 
	\begin{equation*}
	\begin{aligned}
	&    \boldsymbol{D}=\boldsymbol{B}^{\frac{1}{2}}\left(\boldsymbol{I}_{N_r}-\boldsymbol{P}_{\boldsymbol{B}^{\frac{1}{2}}\boldsymbol{W}_{RF}^{\left(K\right)}}\right)\boldsymbol{B}^{\frac{1}{2}},\\
	&\gamma=\text{tr}\left(\boldsymbol{P}_{\boldsymbol{B}^{\frac{1}{2}}\boldsymbol{W}_{RF}^{\left(K\right)}}\boldsymbol{B}^{-\frac{1}{2}}\boldsymbol{A}^{\frac{1}{2}}\boldsymbol{B}^{-\frac{1}{2}}\right),\\ 
	&\boldsymbol{G}=\boldsymbol{B}^{\frac{1}{2}}\boldsymbol{P}_{\boldsymbol{B}^{\frac{1}{2}}\boldsymbol{W}_{RF}^{\left(K\right)}} \boldsymbol{B}^{-\frac{1}{2}}\boldsymbol{A}^{\frac{1}{2}}-\boldsymbol{A}^{\frac{1}{2}},\\
	&\boldsymbol{C}=\gamma\boldsymbol{D}+\boldsymbol{G}\boldsymbol{G}^{*}.
	\end{aligned}
	\end{equation*}
\end{prop}
\begin{proof}
	To prove the proposition, we rely on the following lemma: 
	\begin{lem}\label{thm:rank1updtae}
		%        rank-1 update of inverse of inner product identity (3.2.6 in \cite{brouers_burr_2015}):        
		Let $\tilde{\boldsymbol{W}}=\left[\boldsymbol{W}\;\boldsymbol{w}\right]$ and denote $\boldsymbol{Q}=\left(\boldsymbol{W}^{*}\boldsymbol{W}\right)^{-1}$. Then \cite{Petersen06thematrix}
		\begin{equation*}
		\left(\tilde{\boldsymbol{W}}^{*}\tilde{\boldsymbol{W}}\right)^{-1}=\left[
		\begin{array}{cc}
		\boldsymbol{Q}+\alpha\boldsymbol{Q}\boldsymbol{W}^{*}\boldsymbol{w}\boldsymbol{w}^{*}\boldsymbol{W}\boldsymbol{Q}^{*} & -\alpha\boldsymbol{Q}\boldsymbol{W}^{*}\boldsymbol{w}\\
		-\alpha\boldsymbol{w}^{*}\boldsymbol{W}\boldsymbol{Q}^{*} & \alpha
		\end{array}
		\right],
		\end{equation*}
		with $\alpha=\frac{1}{\boldsymbol{w}^{*}\boldsymbol{w}-\boldsymbol{w}^{*}\boldsymbol{W}\boldsymbol{Q}\boldsymbol{W}^{*}\boldsymbol{w}}$.
	\end{lem}
	Using Lemma \ref{thm:rank1updtae} and straightforward algebraic operations we get equality between the two objectives. The first constraint is identical in both problems. 
	It remains to show that the second constraints of (\ref{eq:VectorProb}) and (\ref{eq:KOpt}) are equivalent. 
	
	Since $\boldsymbol{B}$ is invertible, $\boldsymbol{w}\in\mathcal{R}(\boldsymbol{W}_{RF}^{\left(K\right)})$ if and only if $\boldsymbol{B}^{\frac{1}{2}}\boldsymbol{w}\in\mathcal{R}(\boldsymbol{B}^{\frac{1}{2}}\boldsymbol{W}_{RF}^{\left(K\right)})$. 
	Therefore, if $\boldsymbol{w}\in\mathcal{R}(\boldsymbol{W}_{RF}^{\left(K\right)})$, then
	\begin{flalign*}
	&\boldsymbol{w}^{*}\boldsymbol{D}\boldsymbol{w}=\boldsymbol{w}^{*}\boldsymbol{B}\boldsymbol{w}-\boldsymbol{w}^{*}\boldsymbol{B}^{\frac{1}{2}} \boldsymbol{P}_{\boldsymbol{B}^{\frac{1}{2}}\boldsymbol{W}_{RF}^{\left(K\right)}}\boldsymbol{B}^{\frac{1}{2}}\boldsymbol{w}\\
	&=\boldsymbol{w}^{*}\boldsymbol{B}\boldsymbol{w}-\boldsymbol{w}^{*}\boldsymbol{B}^{\frac{1}{2}} \boldsymbol{B}^{\frac{1}{2}}\boldsymbol{w}=0.    
	\end{flalign*}
	In the other direction, note that $\boldsymbol{D}$ is a non-negative Hermitian matrix, and can be decomposed as $\boldsymbol{D}=\boldsymbol{Q}\boldsymbol{Q}^{*}$ for some $\boldsymbol{Q}$. Thus, $\boldsymbol{w}^{*}\boldsymbol{D}\boldsymbol{w}=0$ if and only if $\boldsymbol{Q}^{*}\boldsymbol{w}=\boldsymbol{0}$. Multiplying both sides of the equation by $\boldsymbol{Q}$ yields $\boldsymbol{D}\boldsymbol{w}=\boldsymbol{B}^{\frac{1}{2}}(\boldsymbol{I}_{N_r}-\boldsymbol{P}_{\boldsymbol{B}^{\frac{1}{2}}\boldsymbol{W}_{RF}^{\left(K\right)}})\boldsymbol{B}^{\frac{1}{2}}\boldsymbol{w}=\boldsymbol{0}$. Since $\boldsymbol{B}$ is invertible, it follows that 
	$(\boldsymbol{I}_{N_r}-\boldsymbol{P}_{\boldsymbol{B}^{\frac{1}{2}}\boldsymbol{W}_{RF}^{\left(K\right)}})\boldsymbol{B}^{\frac{1}{2}}\boldsymbol{w}=\boldsymbol{0}$, or
	$\boldsymbol{B}^{\frac{1}{2}}\boldsymbol{w}\in\mathcal{R}(\boldsymbol{B}^{\frac{1}{2}}\boldsymbol{W}_{RF}^{\left(K\right)})$, which is equivalent to $\boldsymbol{w}\in\mathcal{R}(\boldsymbol{W}_{RF}^{\left(K\right)})$.  
	Hence, we proved that
	\begin{equation*}
	\boldsymbol{w}\in\mathcal{R}\left(\boldsymbol{W}_{RF}^{\left(K\right)}\right)\iff\boldsymbol{w}^{*}\boldsymbol{D}\boldsymbol{w}=0.
	\end{equation*}
	Since $\boldsymbol{D}$ is non-negative definite, $\boldsymbol{w}^{*}\boldsymbol{D}\boldsymbol{w}\geq0$, for all $\boldsymbol{w}$. The previous connection then yields
	\begin{equation*}
	\boldsymbol{w}\notin\mathcal{R}\left(\boldsymbol{W}_{RF}^{\left(K\right)}\right)\iff\boldsymbol{w}^{*}\boldsymbol{D}\boldsymbol{w}>0.
	\end{equation*}        
	As there is equality between the objectives and feasible sets of (\ref{eq:VectorProb}) and (\ref{eq:KOpt}), both problems are equivalent. 
\end{proof}
Problem (\ref{eq:VectorProb}) has several advantages over (\ref{eq:KOpt}): 1) its objective is a scalar ratio and does not involve matrix inversion. 2) this problem is both a constrained ratio-trace, trace-ratio \cite{wang_trace_2007} and a Rayleigh quotient problem, and therefore may be solved using techniques for either of these well known problems. However, these methods require adaptation to fit the additional hardware constraints, which will not be investigated in this work.
%We leave the investigation of approaches for the solution of (\ref{eq:VectorProb}) under hardware constraints to future work.

One simple method for obtaining a (sub-optimal) solution to (\ref{eq:VectorProb}) is searching over a dictionary $\boldsymbol{\bar{W}}$ with columns in $\mathcal{W}^{N_{r}\times 1}$, and calculating the objective value in (\ref{eq:VectorProb}) for each vector. The vector that corresponds to the largest objective is then chosen. This solution has low computational load thanks to the simplicity of the scalar ratio objective. The choice of dictionary depends on the channel model and hardware constraints. For a fully-connected scheme and a sparse mmWave channel as in (\ref{eq:mmWaveChannel}), a steering vector dictionary as used in \cite{ayach_spatially_2014} exploits the channel structure and yields good performance. For a more general model, Gaussian randomization can be used to construct a good dictionary. In this case, $\bar{\boldsymbol{W}}=\boldsymbol{P}_{\mathcal{W}}\left(\boldsymbol{X}\right)$, with $\boldsymbol{P}_{\mathcal{W}}\left(\cdot\right)$ as defined in Section \ref{sec:SignalModel}, and where $\boldsymbol{X}$ is a random dictionary, commonly chosen as a matrix with columns drawn from $\mathcal{CN}\left(\boldsymbol{0},\boldsymbol{W}_{opt}\boldsymbol{W}_{opt}^{*}\right)$, with $\boldsymbol{W}_{opt}$ from (\ref{eq:optimalW}).
In the simulations section, we used the dictionary based solution with Gaussian randomization. 
To comply with the problem's second constraint, it is necessary to remove the selected column from the dictionary at each iteration. 
The other columns are, with high probability, not in $\mathcal{R}(\boldsymbol{W}_{RF}^{\left(K\right)})$ due to the random complex-Gaussian distribution.
%For scheme S\ref{scheme:PhaseSubArrays} the dictionary changes at every iteration since the support of each column at the combiner is different. Hence, the second constraint is assured at each step. 

\begin{algorithm} 
	\caption{GRTM}
	\text{\textbf{Input} effective channel $\bar{\boldsymbol{H}}$, interference correlation $\boldsymbol{R}_{z}$}\\
	\text{\textbf{Output} analog combiner $\boldsymbol{W}_{RF}$}\\
	\text{\textbf{Initialize} $\boldsymbol{A}=\bar{\boldsymbol{H}}\bar{\boldsymbol{H}}^{*}$, $\boldsymbol{B}=\bar{\boldsymbol{H}}\bar{\boldsymbol{H}}^{*}+\boldsymbol{R}_{z}$, $\boldsymbol{C}=\boldsymbol{A},\boldsymbol{D}=\boldsymbol{B}$,}
	\\ \text{$\;\;\;\;\;\;\;\;\;\;\;\;\;\;\;\;\boldsymbol{W}_{RF}^{\left(0\right)}=[]$}\\
	\text{\textbf{For} $K=0:N_{RF}^r-1$:}
	\begin{enumerate}
		\item \text{\textbf{Solve} the base case (\ref{eq:VectorProb}) to obtain $\boldsymbol{w}$}
		\item \text{\textbf{Update} $\boldsymbol{W}_{RF}^{\left(K+1\right)}=\left[\boldsymbol{W}_{RF}^{\left(K\right)}\;\boldsymbol{w}\right]$}
		\item \text{\textbf{Update} }
		\begin{equation*}
		\begin{aligned}
		&\boldsymbol{D}=\boldsymbol{B}^{\frac{1}{2}}\left(\boldsymbol{I}_{N_r}-\boldsymbol{P}_{\boldsymbol{B}^{\frac{1}{2}}\boldsymbol{W}_{RF}^{\left(K+1\right)}}\right)\boldsymbol{B}^{\frac{1}{2}}\\
		&\gamma=\text{tr}\left(\boldsymbol{P}_{\boldsymbol{B}^{\frac{1}{2}}\boldsymbol{W}_{RF}^{\left(K+1\right)}}\boldsymbol{B}^{-\frac{1}{2}}\boldsymbol{A}^{\frac{1}{2}}\boldsymbol{B}^{-\frac{1}{2}}\right)\\ 
		&\boldsymbol{G}=\boldsymbol{B}^{\frac{1}{2}}\boldsymbol{P}_{\boldsymbol{B}^{\frac{1}{2}}\boldsymbol{W}_{RF}^{\left(K+1\right)}} \boldsymbol{B}^{-\frac{1}{2}}\boldsymbol{A}^{\frac{1}{2}}-\boldsymbol{A}^{\frac{1}{2}}\\
				&\boldsymbol{C}=\gamma\boldsymbol{D}+\boldsymbol{G}\boldsymbol{G}^{*}              
		\end{aligned}
		\end{equation*}    
	\end{enumerate}\label{algo:Greedy}
\end{algorithm} 
%    \subsubsection{Base Case Solutions}\label{sec:GreedyBaseCase}

Simulations demonstrate that GRTM has lower MSE than MaGiQ when the number of RF chains increases. In that case, it outperforms other low complexity methods in various SNR, channel models, and hardware scenarios.  

\subsection{Combiner Design for Kronecker Model Channel Estimation}
Problem (\ref{eq:WrfOpt}) arises in other communication problems, so that GRTM can be used to address those problems as well. One such scenario is MIMO channel estimation under a Kronecker model in the interference limited case \cite{spawc_paper,bjornson_framework_2010}.

This problem consists of a single base station (BS) with $N_{BS}$ antennas and $N_{RF}<N_{BS}$ RF chains, and $K$ single-antenna users in a time-synchronized time division duplex (TDD) system. In the uplink channel estimation phase, each user sends $\tau$ training symbols to the BS, during which the channel is assumed to be constant.

The (discrete time) received signal at the BS can be written as
\begin{align}
\boldsymbol{X} & =\boldsymbol{W}_{RF}^{*}\boldsymbol{H}\boldsymbol{S}^{T}+\boldsymbol{Z},\label{eq:KronReceiveSignal}
\end{align}
\normalsize where $\boldsymbol{H}\in\mathbb{C}^{N_{BS}\times K}$ is the desired
channel between the users to the BS, $\boldsymbol{S}\in\mathbb{C}^{\tau\times K}$ is
the pilot sequence matrix, $\boldsymbol{W}_{RF}\in\mathcal{W}^{N_{BS}\times N_{RF}}$ is the analog combiner and $\boldsymbol{Z}$ is the interference. 

Both the channel and interference follow the \textit{doubly correlated Kronecker model} \cite{tulino_impact_2005}
\begin{equation}
\boldsymbol{H}=\boldsymbol{R}_{r}^{\frac{1}{2}}\boldsymbol{\bar{H}}\boldsymbol{R}_{t}^{\frac{1}{2}}\,\,\,,\,\,\,\boldsymbol{Z}=\boldsymbol{R}_{r}^{\frac{1}{2}}\boldsymbol{\bar{Z}}\boldsymbol{Q}^{\frac{1}{2}},\label{eq:KroneckerModel}
\end{equation} 
with $\boldsymbol{R}_{r}\in\mathbb{C}^{N_{bs}\times N_{bs}}$ the BS's receive side correlation matrix, $\boldsymbol{R}_{t}\in\mathbb{C}^{K\times K}$ the transmit side correlation matrix of the users and $\boldsymbol{Q}\in\mathbb{C}^{K\tau\times \tau}$ the transmit side correlation matrix of the interference.
The matrices $\boldsymbol{\bar{H}}^{N_{bs}\times K},\boldsymbol{\bar{Z}}\in\mathbb{C}^{N_{bs}\times \tau}$
have independent entries with i.i.d complex-normal distribution. All correlation matrices are known.  

In the channel estimation phase, the goal is to estimate $\boldsymbol{H}$ from $\boldsymbol{X}$. For this purpose it is desirable to design the analog combiner $\boldsymbol{W}_{RF}$ to minimize the MSE in estimation. As shown in \cite{spawc_paper}, the estimation error of the MMSE estimator for $\boldsymbol{H}$ from $\boldsymbol{X}$ is inversely proportional to
\begin{align}
\mu & \triangleq\mbox{tr}\left(\boldsymbol{W}_{RF}^{*}\boldsymbol{R}_{r}^2\boldsymbol{W}_{RF}\left(\boldsymbol{W}_{RF}^{*}\boldsymbol{R}_{r}\boldsymbol{W}_{RF}\right)^{-1}\right),\label{eq:mu}
\end{align}
which is equal to (\ref{eq:WrfOpt}) with $\boldsymbol{A}=\boldsymbol{R}_{r}^2$ and $\boldsymbol{B}=\boldsymbol{R}_{r}$. 
Hence, all the previously mentioned methods may be used to design $\boldsymbol{W}_{RF}$.

\section{Numerical Experiments}\label{sec:simulations}
We now demonstrate the performance of the suggested beamformers under different hardware constraints. 
Unless mentioned otherwise, the channel model used is the narrow-band mmWave clustered channel defined in (\ref{eq:mmWaveChannel}), with 
$\alpha_{il}\sim\mathcal{NC}\left(0,1\right)$,  $N_{cl}=6$, and $N_{ray}=1$.  
We consider a simple uniform linear array (ULA) with spacing $d=\frac{\lambda}{2}$, where $\lambda$ is the carrier wavelength.
It follows that the steering vectors depend only on the azimuth and are given by
\begin{equation}
\boldsymbol{a}_{q}\left(\phi\right)=\frac{1}{\sqrt{N}}\left[1,e^{j2\pi\sin\left(\phi\right),\cdots,} e^{j2\pi\left(N_{q}-1\right) \sin\left(\phi\right)}\right],
\end{equation}
for $q\in\left\{r,t\right\}$. The angles $\phi_{il}^{t},\phi_{il}^{r}$ are distributed uniformly over the interval $\left[0,2\pi\right)$. 

In all simulations, the number of transmit antennas is $N_t=10$ and the number of receive antennas is $N_r=15$. Unless mentioned otherwise, the interference is a white complex-Gaussian noise with zero mean and variance $\sigma_{z}^2=1$.
For the precoder experiments, we used a fully-digital receiver, i.e. $N_{RF}^r=N_r$, and for the combiner experiments we used a fully-digital transmitter, i.e. $N_{RF}^t=N_t$.
% We begin by demonstrating the precoder design. For this 

\subsection{Hybrid Precoder Performance}
First, we present the performance of different hybrid precoders. Here, we compare 5 design methods. The first two are the MO\_AltMin algorithm from \cite{yu_alternating_2016}, which performs manifold optimization to minimize (\ref{eq:approxF}) over $\boldsymbol{F}_{RF}$ and $\boldsymbol{F}_{BB}$, and the PE\_AltMin method, also from \cite{yu_alternating_2016}, that assumes $\boldsymbol{F}_{BB}$ to be a scaled orthogonal matrix and employs an algorithm similar to Algorithm \ref{algo:SimpleQuan}, only with $\boldsymbol{F}_{BB}$ calculated in the last step as $\boldsymbol{F}_{BB}=\frac{\sqrt{N_{S}}}{\Vert\boldsymbol{F}_{RF}\boldsymbol{T}\Vert_{F}}\boldsymbol{T}$. 
The third is the SOMP of \cite{ayach_spatially_2014}, that reformulates (\ref{eq:approxF}) as a sparsity problem and chooses the columns of $\boldsymbol{F}_{RF}$ from a dictionary of candidate vectors, and simultaneously solves the corresponding weights $\boldsymbol{F}_{BB}$. 
Last two are our methods, the MaGiQ in Algorithm~\ref{algo:SimpleQuan} and the Alt-MaG method in Algorithm \ref{algo:IterQuan} with the MO\_AltMin solution used to obtain $\boldsymbol{F}_{RF},\boldsymbol{F}_{BB}$ in step~1.
For the AltMin algorithms, we used the Matlab's functions \cite{AltMin}. The dictionary used for the SOMP algorithm consists of 1000 steering vectors $\left\{\boldsymbol{a}_t\left(\frac{2\pi}{1000}q\right)\right\}_{q=1}^{1000}$. As performance measure we consider the gap between the estimation error of the optimal precoder (\ref{eq:Fopt}) and the hybrid one, defined by $\epsilon-\epsilon_{opt}$ with $\epsilon=\frac{1}{N_{s}\cdot Q}\sum_{q=1}^{Q}\Vert \hat{\boldsymbol{s}}-\boldsymbol{s}\Vert^2$ and $Q$ the number of realizations.

We begin by comparing MaGiQ and PE\_AltMin considering a fully-connected phase shifters network S\ref{scheme:FullyPhaseOnly}. As mentioned before, these two algorithms are very similar except for the last step of calculating $\boldsymbol{F}_{BB}$ which is suboptimal in PE\_AltMin. This is demonstrated in Fig.~\ref{fig:PEsim} where MaGiQ's lower MSE can be seen. In the next simulations, PE\_AltMin will not be tested.

%Next, we consider the fully-connected phase shifters network S\ref{scheme:FullyPhaseOnly}. 
Figure \ref{fig:precoder_err} presents the estimation gap with respect to the number of RF chains for the different design methods. The SOMP algorithm suffers from a large performance gap, especially when the number of RF chains is much less than the number of multipath components, i.e. $N_{RF}^t\ll N_c$. MaGiQ enjoys both low complexity and good performance due to the alternating minimization over $\boldsymbol{T}$ in (\ref{eq:OptTF}). The MO-AltMin algorithm achieves low MSE at the cost of long running time. The gap between Alt-MaG and MO\_AltMin demonstrates the potential gain in optimizing over the matrix $\boldsymbol{T}$ in (\ref{eq:OptTF}), in comparison to optimizing only $\boldsymbol{F}_{RF},\boldsymbol{F}_{BB}$ as in (\ref{eq:approxF}). This approach substantially improves the performance, with negligible increase in complexity. 

\begin{figure}
	\centering
	\vspace{-0.45cm}
	\includegraphics[width = 0.5\textwidth]{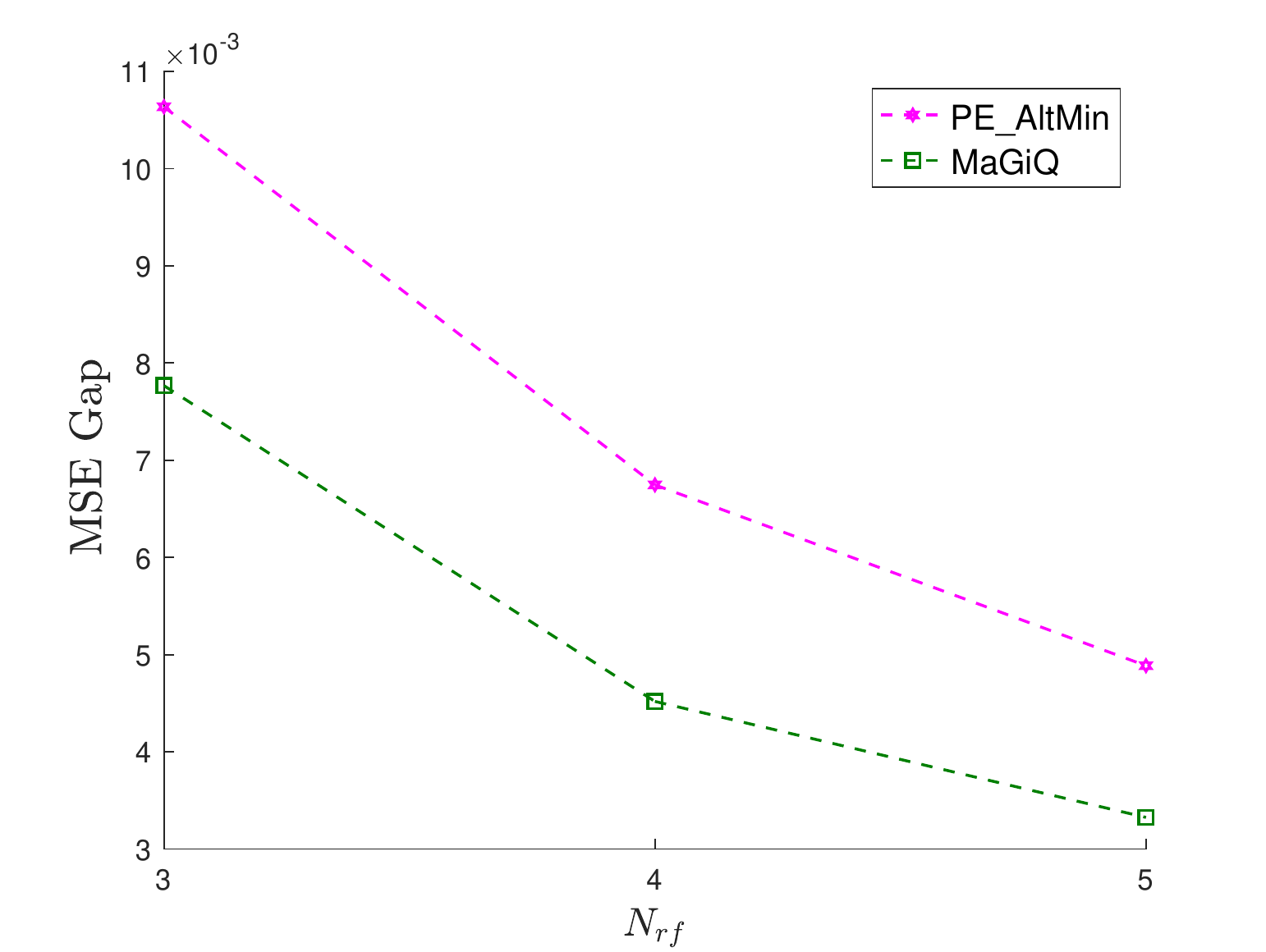}
	\vspace{-0.5cm}
	\protect\caption{Precoder estimation gap $\epsilon-\epsilon_{opt}$ vs. number of RF chains $N_{RF}^t$, for MAGiQ and PE\_AltMin.}\label{fig:PEsim}
	\vspace{0.3cm}
\end{figure}
\begin{figure}
	\centering
	\vspace{-0.5cm}
	\includegraphics[width = 0.5\textwidth]{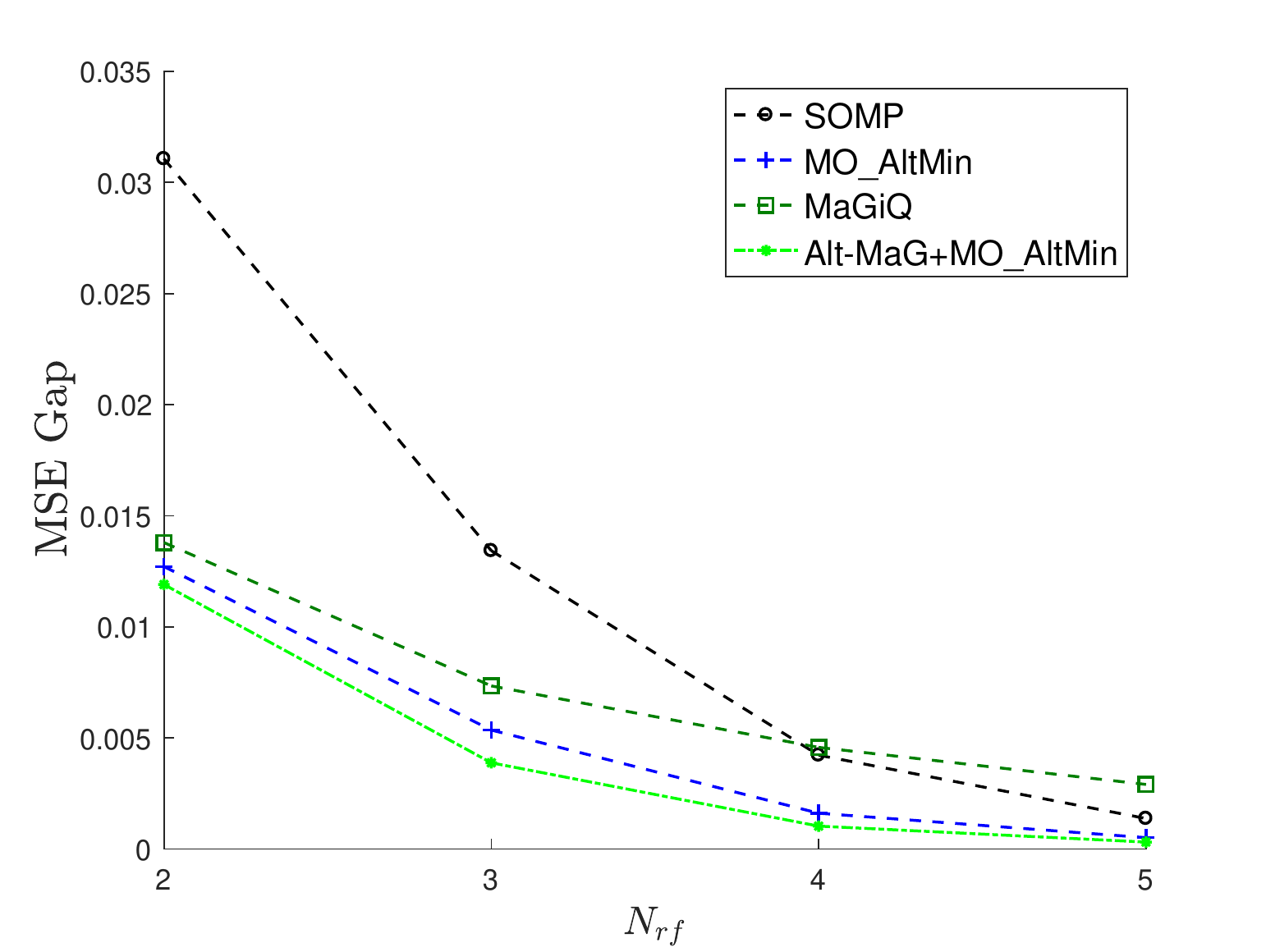}
	\vspace{-0.5cm}
	\protect\caption{Precoder estimation gap $\epsilon-\epsilon_{opt}$ vs. number of RF chains $N_{RF}^t$, for a fully-connected phase shifters network S\ref{scheme:FullyFirst}. }\label{fig:precoder_err}
	%        \vspace{-0.1cm}
\end{figure}

\subsection{Hybrid Combiner Performance}
We now investigate the combiner's performance. 
Here, an additional algorithm is tested, the GRTM in Algorithm~\ref{algo:Greedy}, with the dictionary-based method in Section \ref{sec:GreedyAlg} as the base case solution. 
In contrast to the precoder case, here the SOMP algorithm solves the problem (\ref{eq:HeathProb}) as suggested in \cite{ayach_spatially_2014}. 

Figure \ref{fig:comb_err_vs_rf} presents the estimation gap from the optimal combiner for all approaches with the fully-connected phase shifters scheme S\ref{scheme:FullyPhaseOnly}. It can be seen that MO-AltMin provides the best results. However, its runtime is an order-of-magnitude larger than the runtime of the other methods. Again, in the case of very few RF chains, MaGiQ is preferred over the greedy methods GRTM and SOMP, which suffer from large performance gap in that scenario. With additional RF chains, the greedy algorithms outperform the simple MaGiQ, with a slight advantage to GRTM.

Next, we show that in some special cases, MaGiQ coincides with the fully-digital solution and outperforms other methods. 
%We demonstrate this using the scenario described in the example specific practical. 
%The first case is a theoretical scenario where the conditions of Proposition~\ref{prop:MaGiQ} hold. Here, $\boldsymbol{H}$ is a general random channel with unimodular singular vectors, constructed by drawing a matrix with independent complex-Gaussian entries, and projecting its left singular vectors onto the unimodular set. Figure \ref{fig:PrivateCase} shows the estimation error $\epsilon$ of the different algorithms.  It can be seen that MaGiQ's performance indeed equals that of the optimal fully-digital combiner. MO-AltMin enjoys flexibility in the combiner vectors structure, and hence has better performance than the dictionary-limited methods, but fails to converge to the optimal solution, even though it is on the unimodular manifold. The dictionary-based methods suffer from a large performance gap due to the dictionary constraint. We used SOMP with a steering vectors dictionary and GRTM with a general unimodular one. As expected, the steering dictionary performs worse than the general unimodular one.
%
We demonstrate this using the scenario from the asymptotic example given in Section~\ref{sec:Precoder}. Here, $\boldsymbol{H}$ is the mmWave channel (\ref{eq:mmWaveChannel}), with $N_{c}=4$, $N_{RF}^r=4$, and $N_r=150$. In this case, MO-AltMin cannot be used due to its long runtime that grows with the number of antennas. Since $N_{c}\ll N_r$ the asymptotic analysis holds, and $\boldsymbol{H}$ has unimodular singular vectors. Figure~\ref{fig:PrivateCasePractical} shows the estimation error $\epsilon$ of the different algorithms. It can be seen that MaGiQ's performance coincides with the fully-digital combiner. We used a steering dictionary for both SOMP and GRTM.  Since the singular vectors of the channel also have steering structure, the greedy methods performance gap is small, but still exists. 

Next, we investigate the partially connected schemes.  Figure~\ref{fig:comb_err_vs_snr} shows the performance of different combiners for the fixed sub-arrays network S\ref{scheme:PhaseSubArrays} and for the flexible one S\ref{scheme:FlexiLast}. For both scenarios, we set $G=5$. We first notice the additional performance gain of the flexible architecture. This result is natural since the additional switches allow each RF chain to choose the antennas that contribute to it, in contrast to the fixed sub-arrays case, that enforces each RF chain to be connected to a small predetermined array. In both scenarios, the greedy methods outperform MaGiQ, but while SOMP yields better performance for the fixed arrays, when adding switches to allow for flexibility, GRTM is preferred.

The fully-connected network S\ref{scheme:FullyFirst}, that involves switches as well as phase shifters is considered next. Here, we use a general channel model such that $\boldsymbol{H}$ has i.i.d complex-Gaussian entries and the interference $\boldsymbol{z}$ is now a complex-Gaussian vector with arbitrary full rank covariance matrix $\boldsymbol{R}_{z}$. For the SOMP and GRTM algorithms, we used the dictionary $\bar{\boldsymbol{W}}=\boldsymbol{P}_{\mathcal{W}}\left(\boldsymbol{X}\right)$, with $\mathcal{W}$ the feasible set for scheme S\ref{scheme:FullyFirst}, as explained in Section \ref{sec:GreedyAlg}. Figure \ref{fig:general_model} presents the MSE $\epsilon$ as a function of SNR. The complex MO-AltMin algorithm now achieves approximately the same performance as the low complexity GRTM. This is because MO-AltMin cannot use the additional flexibility of the switches and produces only unimodular beamformer vectors. This demonstrates a trade-off between hardware and computational complexity: one may use the simpler architecture S\ref{scheme:FullyPhaseOnly} while achieving the same performance as the complex network S\ref{scheme:FullyFirst}, with the cost of using a heavy computational algorithm such as MO-AltMin. It can be further noticed that SOMP's performance falls short in comparison to others, possibly due to the different interference structure. 
\begin{figure}
	\centering
	\vspace{-0.4cm}
	\includegraphics[width = 0.5\textwidth]{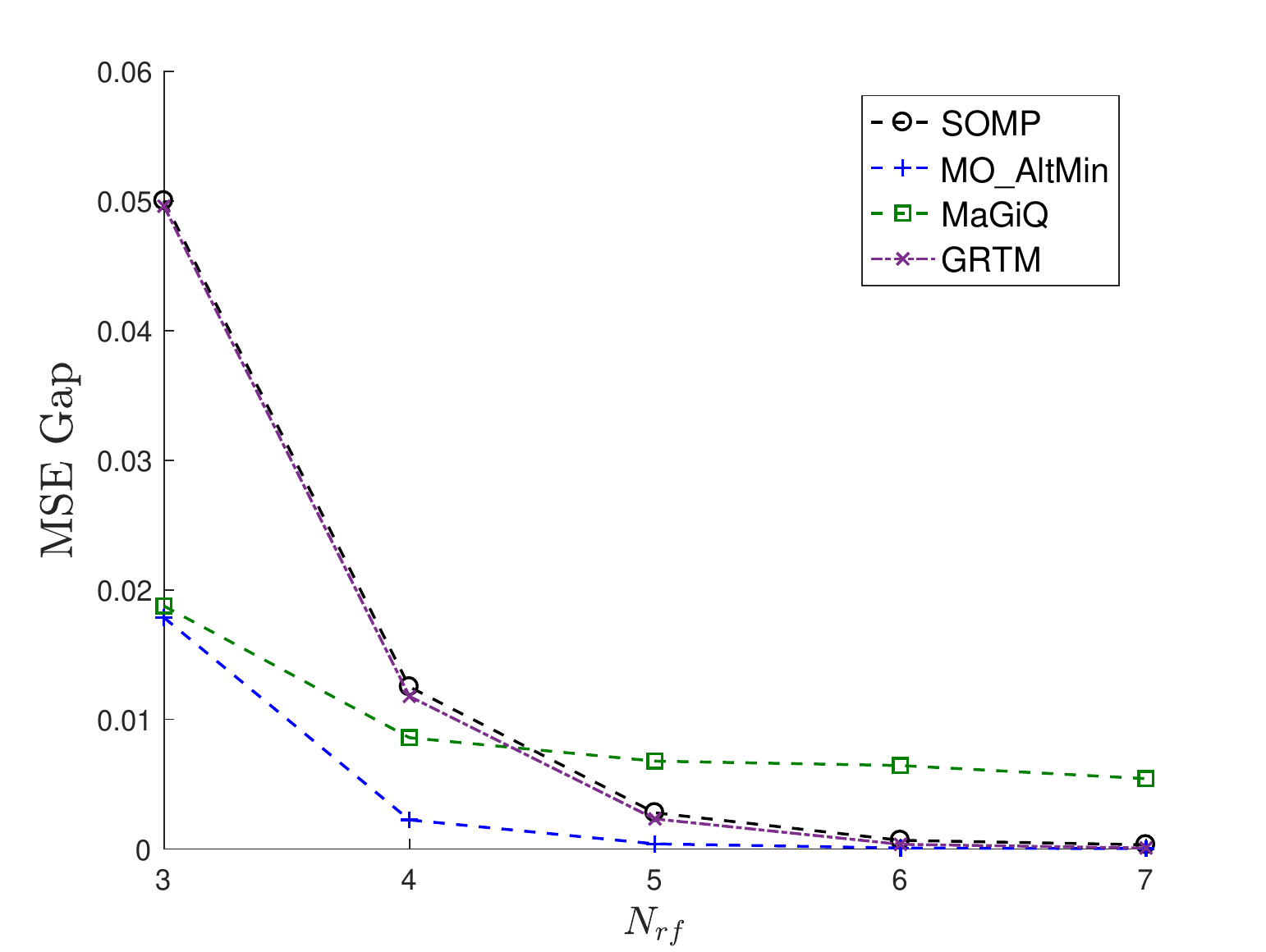}
	\vspace{-0.5cm}
	\protect\caption{Combiner estimation gap $\epsilon-\epsilon_{opt}$ vs.  number of RF chains $N_{RF}^t$, for a fully-connected phase shifters network S\ref{scheme:FullyFirst}. 
		\label{fig:comb_err_vs_rf}}
	\vspace{0.1cm}
\end{figure}
%\begin{figure}
%	\centering
%	\vspace{-0.5cm}
%	\includegraphics[width = 0.5\textwidth]{SimResults/CombinerPrivateCase2}
%	\vspace{-0.5cm}
%	\protect\caption{Combiner estimation error vs. SNR for a fully-connected phase shifters network S\ref{scheme:FullyPhaseOnly} and channel $\boldsymbol{H}$ with unimodular singular vectors. 
%		\label{fig:PrivateCase}}
%	\vspace{0.1cm}
%\end{figure}
%\vspace{-1cm}
\begin{figure}
	\centering
	\vspace{-0.6cm}
	\includegraphics[width = 0.5\textwidth]{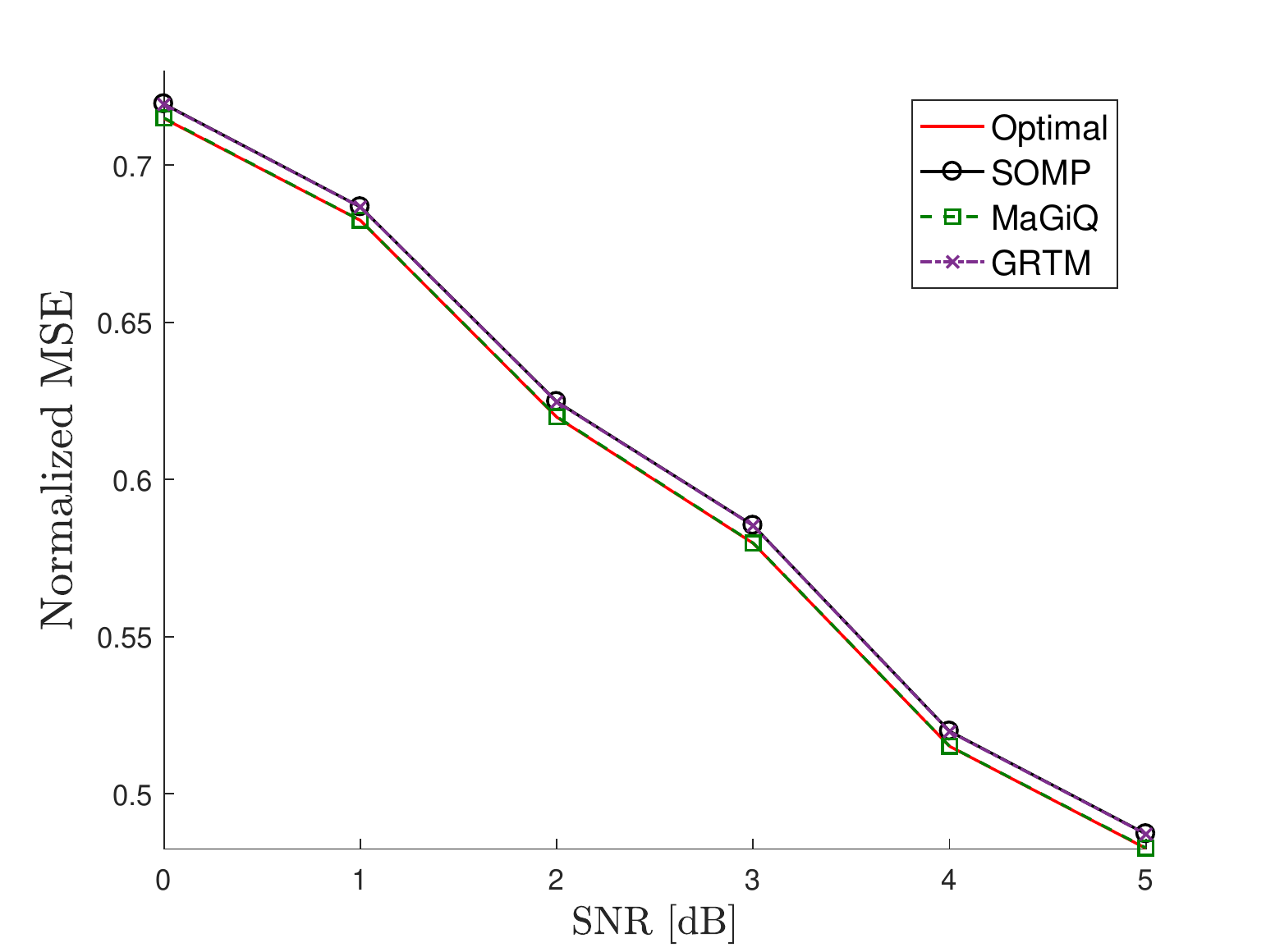}
	\vspace{-0.5cm}
	\protect\caption{Combiner estimation error vs. SNR for a fully-connected phase shifters network S\ref{scheme:FullyPhaseOnly} and mmWave channel with asymptotic number of antennas. 
		\label{fig:PrivateCasePractical}}
	\vspace{0cm}
\end{figure}
\begin{figure}
	\centering
	\vspace{-0.6cm}
	\includegraphics[width = 0.5\textwidth]{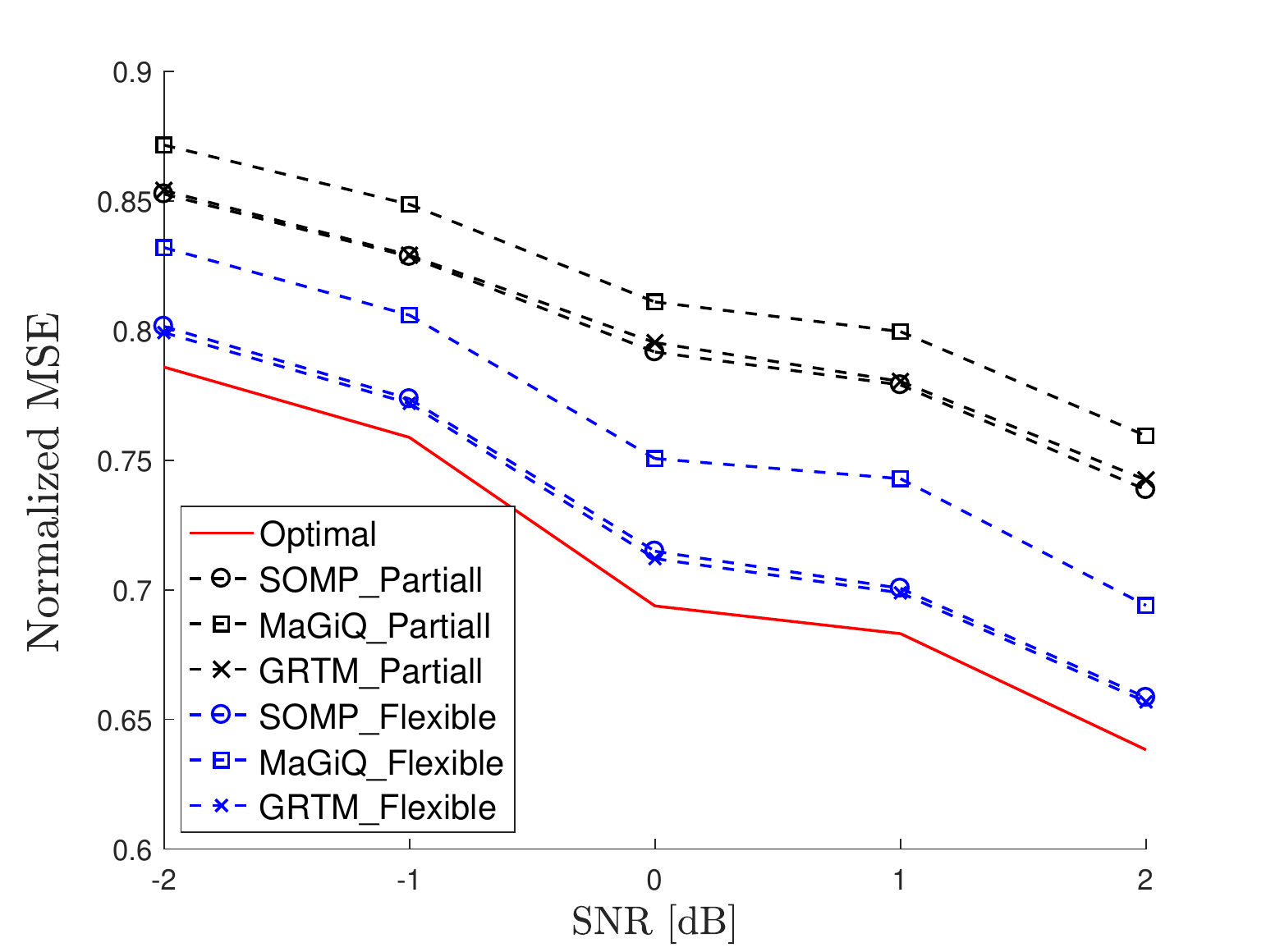}
	\vspace{-0.5cm}
	\protect\caption{Combiner estimation error vs. SNR for the partially connected phase shifters networks S\ref{scheme:PhaseSubArrays} and S\ref{scheme:FlexiLast} with $G=5$. 
		\label{fig:comb_err_vs_snr}}
	\vspace{0.3cm}
\end{figure}
\begin{figure}
	\centering
	\vspace{-0.95cm}
	\includegraphics[width = 0.5\textwidth]{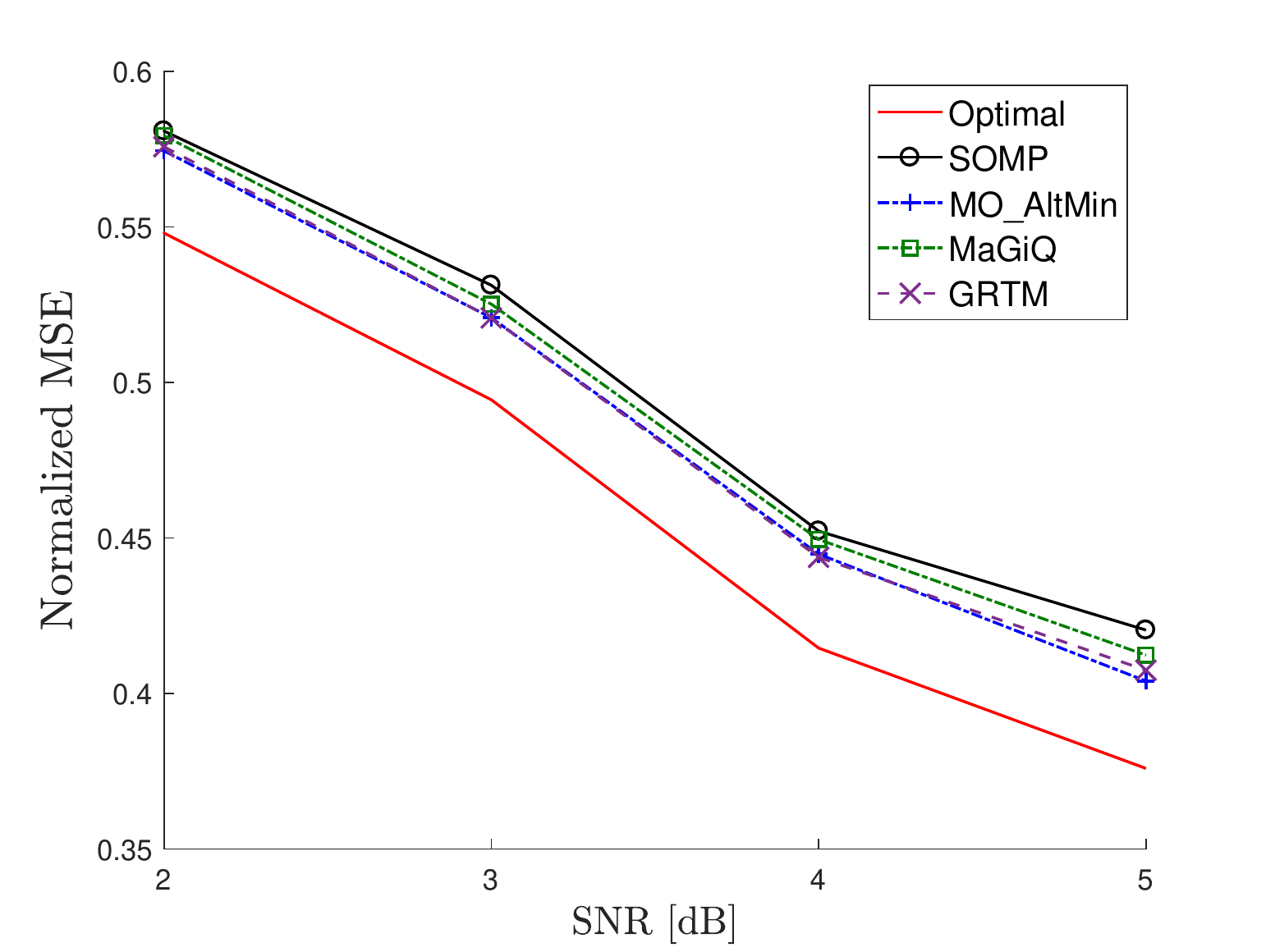}
	\vspace{-0.5cm}
	\protect\caption{Combiner estimation error vs. SNR for the fully-connected phase shifters and switches network S\ref{scheme:FullyFirst} with general channel model and interference. 
		\label{fig:general_model}}
	\vspace{0.1cm}
\end{figure}
\begin{figure}
	\centering
	\vspace{-0.45cm}
	\includegraphics[width = 0.5\textwidth]{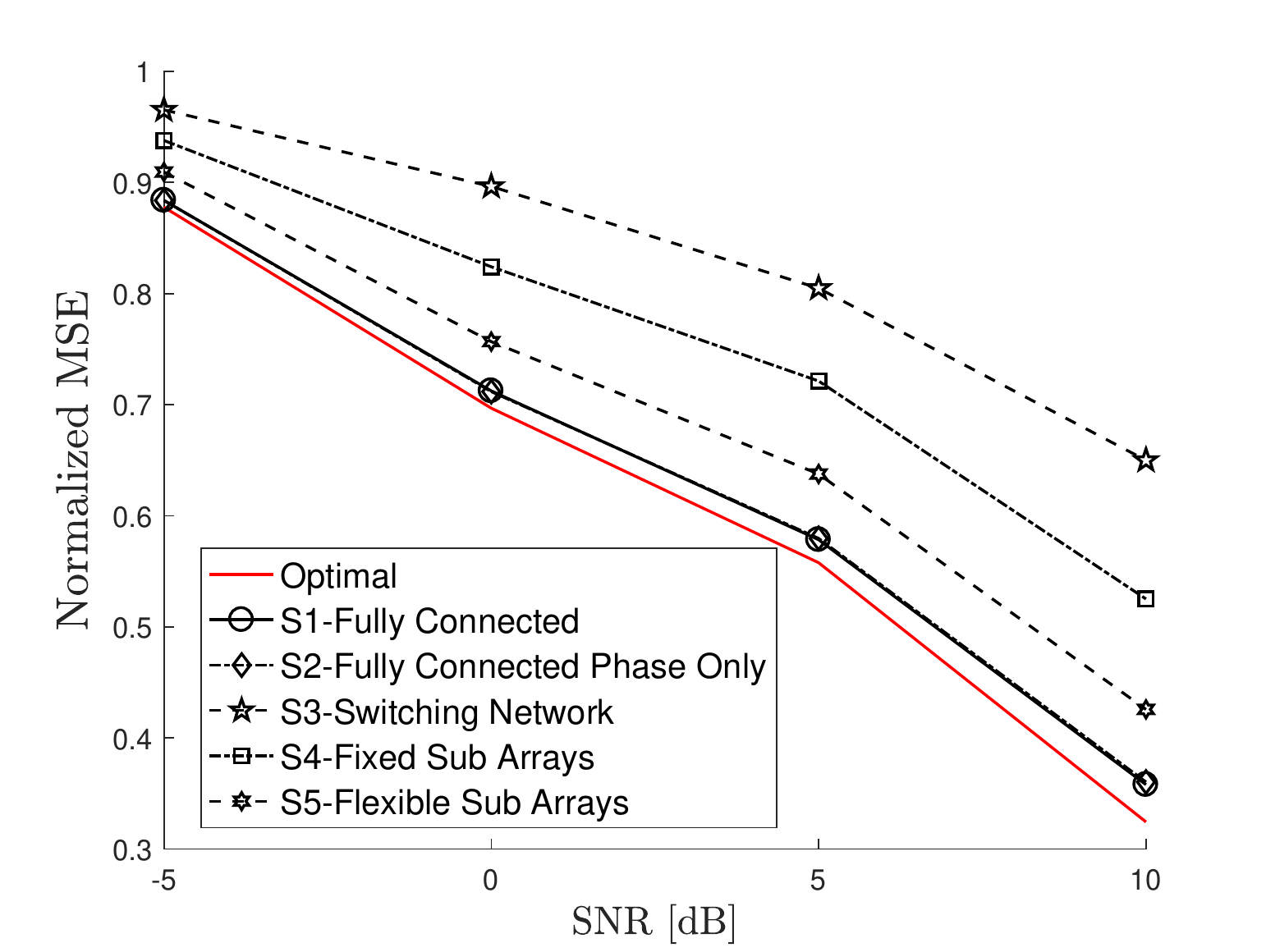}
	\vspace{-0.5cm}
	\protect\caption{Combiner estimation error vs. SNR for a partially connected phase shifters network S\ref{scheme:PhaseSubArrays} with $G=3$. 
		\label{fig:schemes_vs_snr}}
	%        \vspace{0.3cm}
\end{figure}

\subsection{Hardware Schemes Comparison}
The last simulation considered the different hardware choices. 
In Fig.~\ref{fig:schemes_vs_snr}, the performance of the MaGiQ combiner under different hardware constraints is demonstrated. As expected, the simple switching network S\ref{scheme:Switching} suffers from the largest MSE, as its connectivity is the most limited. The two partially connected cases S\ref{scheme:PhaseSubArrays} and S\ref{scheme:FlexiLast} with $G=5$ offer lower MSE, due to the additional hardware complexity, where the second performs better as a result of the additional switches. The two fully-connected architectures S\ref{scheme:FullyFirst} and S\ref{scheme:FullyPhaseOnly} show similar performance, suggesting that the additional switches in the first offer negligible improvement. This surprising result can be explained by the channel model which implies that the optimal combiner consists of sums of steering vectors, that can be efficiently described using unimodular vectors. 

\section{Conclusions}\label{sec:conclusions}

We developed a framework for hybrid precoder and combiner design, suitable for various channel models and hardware settings. 
We considered a data estimation problem and aimed at minimizing the estimation MSE over all possible hybrid precoders/combiners. 
For the precoder side, we suggested a family of iterative algorithms, Alt-MaG, that approximates the optimal fully-digital precoder while seeking the solution in the optimal set that results in the smallest approximation gap from the hybrid precoder of the previous iteration. The potential gain in exploiting the unitary degree of freedom in the fully-digital precoder was demonstrated in simulations. We further suggested a simple low complexity algorithm, MaGiQ, that achieves good approximation using a simple quantization function, and yields better performance than other methods suggested in the literature. We also showed that in some special cases it coincides with the optimal fully-digital solution.  

Next, we adjusted the MaGiQ algorithm so that it can be used for combiner design. We then suggested an additional greedy algorithm, GRTM, that directly minimizes the MSE using a simple scalar-ratio objective at each iteration. GRTM achieves lower MSE than MaGiQ when the number of RF chains increases. Experimental results showed that our low-complexity algorithms enjoy good performance in various scenarios.

Finally, using simulations, we showed that in partially connected schemes and sparse multipath mmWave channels, adding switches to the analog network offers a large increase in performance, while in fully-connected cases the improvement is negligible.

\small
\bibliographystyle{IEEEtran}
\bibliography{CombinerBib}

\end{document}